\documentclass[aps,pre,twocolumn,superscriptaddress,10pt]{revtex4-2}
\usepackage[utf8]{inputenc}
\usepackage{multirow,amsthm,amssymb,amsbsy,amsmath,epsfig,epstopdf,bm,float}
\usepackage{bm}
\usepackage{graphicx}
\usepackage{booktabs}
\usepackage{verbatim}
\usepackage{dcolumn}
\usepackage{color}
\usepackage{xcolor}
\usepackage[normalem]{ulem}
\usepackage{soul}
\usepackage{braket}
\usepackage{subfig}%para poner varias figuras
\usepackage[export]{adjustbox}
\usepackage{enumerate}

\newcommand{\bb}[1]{\mathbb{#1}}
\newcommand{\tr}{\text{tr}}
\newcommand{\vertiii}[1]{{\left\vert\kern-0.25ex\left\vert\kern-0.25ex\left\vert #1 
    \right\vert\kern-0.25ex\right\vert\kern-0.25ex\right\vert}}

\newtheorem{theorem}{Theorem} %NO muestra la section

\newtheorem{remark}[theorem]{Remark}
\newtheorem{proposition}[theorem]{Proposition}

\newtheorem{example}[theorem]{Example}

\begin{document}
\title{Quantum reservoir computing in finite dimensions}
\author{Rodrigo Mart\'inez-Pe\~na}
\email{rmartinez@ifisc.uib-csic.es}
\affiliation{Instituto de F\'{i}sica Interdisciplinar y Sistemas Complejos (IFISC, UIB-CSIC), Campus Universitat de les Illes Balears, 07122 Palma de Mallorca, Spain}
\author{Juan-Pablo Ortega}
\email{juan-pablo.ortega@ntu.edu.sg}
\affiliation{Division of Mathematical Sciences, Nanyang Technological University, 21 Nanyang Link, Singapore 637371}

\date{ \today }

\begin{abstract}
Most existing results in the analysis of quantum reservoir computing (QRC) systems with classical inputs have been obtained using the density matrix formalism. This paper shows that alternative representations can provide better insights when dealing with design and assessment questions. More explicitly, system isomorphisms are established that unify the density matrix approach to QRC with the representation in the space of observables using Bloch vectors associated with Gell-Mann bases. It is shown that these vector representations yield state-affine systems (SAS) previously introduced in the classical reservoir computing literature and for which numerous theoretical results have been established. This connection is used to show that various statements in relation to the fading memory (FMP) and the echo state (ESP) properties are independent of the representation, and also to shed some light on fundamental questions in QRC theory in finite dimensions. In particular, a necessary and sufficient condition for the ESP and FMP to hold is formulated using standard hypotheses, and contractive quantum channels that have exclusively trivial semi-infinite solutions are characterized in terms of the existence of input-independent fixed points. 
\end{abstract}


\maketitle

\section{Introduction}
The development of noisy intermediate-scale quantum (NISQ) devices is attracting a great deal of attention from the quantum community. Recent advancements in fields such as quantum computation \cite{arute2019quantum}, quantum simulation \cite{choi2016exploring}, and quantum communications \cite{chen2021integrated} are just examples of the prosperous future that awaits these technologies. Nevertheless, NISQ devices are already demonstrating in the meantime that they can be very useful for diverse research fields like physics, chemistry, and optimization \cite{bharti2022noisy}, even providing quantum advantage \cite{arute2019quantum,zhong2020quantum,madsen2022quantum}. Machine learning (ML) is another example of thriving synergies with NISQ technologies. In this context, {\it quantum machine learning} (QML) aims to exploit the specific features of quantum mechanics to obtain an advantage over its classical counterparts when dealing with machine learning tasks, both with classical and quantum data \cite{biamonte2017quantum}. There is already positive evidence in this direction, both from a theoretical point of view in the fault-tolerant picture \cite{liu2021rigorous} and in experiments \cite{huang2022quantum}.   

The flexibility and range of action of ML and QML techniques is typically specified by the so-called {\it universality approximation theorems}. A universal approximation property takes place when a proposed restricted family of functions can approximate any function in a much larger class with arbitrary precision. There are many results of this type that are part of classical analysis dealing with, for instance, polynomials and Fourier series, and others that were added in the early days of ML like, for example, feed-forward neural networks \cite{cybenko1989approximation,hornik1989multilayer,hornik1991approximation}. Further results of this type have been proved for various ML paradigms like recurrent neural networks \cite{schafer2006recurrent}, support vector machines \cite{hammer2003note}, extreme learning machines \cite{huang2006universal, RC12}, or kernel methods \cite{steinwart2001influence, micchelli2006universal, RC25}. Universality results have also been obtained in the QML context, such as for one qubit algorithms \cite{perez2020data,perez2021one}, and general quantum circuits \cite{schuld2021effect,goto2021universal}. A framework in which we are particularly interested is {\it quantum reservoir computing} (QRC). As in classical reservoir computing (RC) \cite{jaeger2001echo, lukosevicius, tanaka:review}, QRC harnesses the rich dynamics of (quantum) dynamical systems to solve tasks where memory and prediction capabilities are required. Examples of application of these techniques are found in the prediction of chaotic time-series \cite{jaeger2004harnessing, GHLO2014, shahi2022prediction} and complex spatiotemporal dynamics \cite{pathak2018, vlachas2020backpropagation,tran2020higher, wikner2021using}. Since the first work on QRC \cite{fujii2017harnessing} many have followed (see \cite{ghosh2021quantum,mujal2021opportunities} for reviews). This includes both experimental implementations \cite{chen2020temporal,dasgupta2020designing,suzuki2022natural,kubota2022quantum,molteni2022optimization} as well as theoretical contributions on the universal approximation question \cite{chen2019learning,chen2020temporal,nokkala2021gaussian}. The latter was inspired by the works in the classical framework \cite{grigoryeva2018echo,grigoryeva2018universal,gonon2019reservoir, RC20}, where the discrete-time setting of RC theory fits well within the quantum dynamical map description.

The universal approximation property in the RC context brings to the table some conditions that dynamical systems should meet to ensure it. The most prominent ones are the {\it echo state property} (ESP) \cite{jaeger2001echo} and the {\it fading memory property} (FMP) \cite{boyd1985fading}.  A system has the FMP if inputs that are close in the recent past produce outputs that are also close, independently of what happened in the distant past. The ESP guarantees that a well-defined input/output map can be associated with our system and amounts to an existence and uniqueness property with respect to the input sequence that is fed into the system. The ESP and FMP are standard requirements in many stochastic and deterministic learning paradigms since they mathematically encode the asymptotic decorrelation (and even independence) between physical states and initial conditions that most physical systems exhibit as time goes by. There are many physical mechanisms that lead to the declining relevance of a given initial condition in a system state as the temporal distance between them increases. For instance, two pervasive phenomena in applications in this direction are chaos (high sensitivity to initial conditions) and dissipation.  Fading memory is an important modeling feature when using physical systems because a reservoir system can only store a finite amount of information in its trainable parameters \cite{wright2020capacity}; this implies that information must be erased as time goes by in order to store new input information. Under very mild mathematical conditions, ESP and FMP are equivalent to the input-forgetting property \cite{RC9} which mathematically encodes the information removal process that is needed to learn the newly fed one.

The ESP property has been studied in the context of QRC in different works. Nurdin and Chen \cite{chen2019learning} provided sufficient conditions for the ESP (and FMP) to hold in terms of the contractivity of the quantum map acting on a restricted domain  \cite{chen2019learning,chen2020temporal,chen2022nonlinear}, while Tran and Nakajima also define the ESP using contractivity of the quantum map but without considering restrictions. They also numerically connect the ESP with the spectra of the quantum maps \cite{tran2020higher,tran2021learning}. Both approaches describe the ESP in the density matrix language, that is, the dynamical equations are defined in the space of quantum states. However, this description might be limited somehow, because, as we will see later on, more information can be extracted if we choose a description in terms of observables.

This paper aims to fill some gaps in the description of QRC with finite dimensional systems and classical inputs. More precisely, we first unify the quantum ESP defined in previous works \cite{chen2019learning,chen2020temporal,chen2022nonlinear,tran2020higher,tran2021learning} using the norm of quantum maps, with the classical notion of ESP using observables, and show that they are equivalent for finite-dimensional quantum systems. We extend in passing these results to the FMP.  More specifically, we establish explicit system isomorphisms between the description of QRC systems using the space of density matrices as state space and the one that uses the space of Bloch vectors associated to Gell-Mann bases. That choice of basis, which is customary in the study of quantum systems, happens to yield non-homogeneous affine state dynamics (the corresponding systems are called {\it state-affine} (SAS)) of the type introduced in \cite{grigoryeva2018universal}. The connection between QRC and  SAS systems has important consequences.
Our results in that context are contained in Proposition \ref{prop:sufficient} and in Theorems \ref{th:unital} and \ref{th:constant}. The proposition presents a necessary and sufficient condition for the ESP and FMP to hold in different representations with just a few requirements, such as the compactness of the input space. We emphasize that this hypothesis is satisfied in most  RC tasks and when dealing with implementations of RC systems with dedicated hardware since the experimental ranges of the physical systems involved are always finite \cite{tanaka:review}. The theorems exhibit common situations that should be avoided in the design of quantum channels so that fully operational QRC systems are obtained. We shall work in an idealized framework in which observables are obtained after an infinite number of measurements, with no statistical error. Even in such an idealized setting, we expect that our results can contribute to the general understanding of QRC experiments in finite dimensional systems, as they have already been implemented in \cite{chen2020temporal,dasgupta2020designing,kubota2022quantum,molteni2022optimization}.

The structure of the paper is as follows. Section \ref{sec:definitions} introduces the general framework and the definitions that will be needed along the paper. Definitions of the spaces of operators and quantum maps are included in Section \ref{sec:q_definitions}, while all the RC ingredients will be presented in Section \ref{sec:rc_definitions}, together with some preliminary results. The main results are contained in Section \ref{sec:results}. Section  \ref{sec:discussion} includes a brief discussion on some of the consequences of Theorems \ref{th:unital} and \ref{th:constant}, and Section \ref{sec:conclusions} concludes the paper.

\section{Definitions}\label{sec:definitions}

\subsection{Quantum definitions}\label{sec:q_definitions}

We start by defining the space of quantum systems. See, for example, \cite{redei2007quantum,bengtsson2017geometry} for further details. Consider a complex Hilbert space $\mathcal{H}$. The set of all {\it bounded operators} $\mathcal{B}(\mathcal{H})$ that act on $\mathcal{H}$ is a complex vector space under point-wise addition and scalar multiplication, and it forms an algebra under composition. If we add the involution $A\rightarrow A^{\dagger}$ given by the adjoint operation, then $\mathcal{B}(\mathcal{H})$ is also a  $C^*$-algebra with respect to the {\it operator norm} $\vertiii{\cdot }_{op}$ defined by 
\begin{equation}
   \vertiii{A}_{op} :=\sup_{||\psi||=1}\{||A\psi||,\ \psi\in \mathcal{H} \}, 
\end{equation}
where $A\in \mathcal{B}(\mathcal{H})$. Along this manuscript, we will denote all the induced operator and matrix norms with the symbol $\vertiii{\cdot }$.  The predual of $\mathcal{B}(\mathcal{H})$ is the Banach space $\mathcal{T}(\mathcal{H})$ of all {\it trace-class operators} on $\mathcal{H}$ that have finite {\it trace norm}  $||A||_1:=\tr\sqrt{AA^{\dagger}}$. The trace norm is a particular case (with $p=1 $) of the {\it Schatten norms} defined by $||A||_p:=\left(\tr\left(\left(\sqrt{AA^{\dagger}}\right) ^p\right)\right)^{1/p}$.
The space of quantum {\it density matrices} $\mathcal{S}(\mathcal{H})$ is a compact convex subset (see Section \ref{sec:rc_definitions} and the argument above \eqref{inclusions spaces w}) of the normed vector space $\mathcal{T}(\mathcal{H})$ defined by 
\begin{equation}
    \mathcal{S}(\mathcal{H})=\{\rho\in\mathcal{T}(\mathcal{H})\ |\ \rho^{\dagger}=\rho,\ \rho\geq 0,\ \tr(\rho)=1\}.
\end{equation} 
As $\mathcal{S}(\mathcal{H})$ is a closed subset of the Banach space $\mathcal{T}(\mathcal{H})$, it is then a complete metric space when using the distance induced by $\left\|\cdot \right\|_1$. All along this paper, we restrict ourselves to finite-dimensional Hilbert spaces, for which the spaces of bounded and trace class operators coincide, and we shall use the symbol $\mathcal{B}(\mathcal{H})$ to refer to both of them. Moreover, in that case, $\mathcal{S}(\mathcal{H})$ is a complete metric space with respect to any norm since all the norms are equivalent. We shall reserve the symbol $d \in {\bb N} $ for the dimension of ${\mathcal H}$.

We now introduce the notion of {\it quantum channel}. All definitions and further properties of these maps can be found in \cite{raginsky2002dynamical,burgarth2007generalized,burgarth2013ergodic,wolf2012quantum} and references therein.
A quantum channel is a linear map $T: \mathcal{B}(\mathcal{H})\rightarrow \mathcal{B}(\mathcal{H})$ that is {\it completely positive and trace preserving} (CPTP). We recall here that a {\it trace preserving} map $T$ is the one that satisfies that $\tr(T(A))=\tr(A)$ for any $A \in \mathcal{B}({\mathcal H})$. Moreover, we say that $T$ is {\it positive} when it maps positive semi-definite operators to positive semi-definite operators. Finally, {\it completely positive} maps $T$ are those positive maps which, when extended to a larger space using the tensor map $T\otimes I_k$, with $I_k$ the identity map in dimension $k$, they also yield a positive map for any $k \in \mathbb{N}$. A linear map is CPTP if and only if it is possible to find a {\it Kraus decomposition} associated to a set of operators $\{K_i\}_{i\in X}$ such that for all $A\in \mathcal{B}(\mathcal{H})$:
\begin{equation}
\label{kraus decomposition}
    T(A)=\sum_{i\in X}K_iAK^{\dagger}_i,
\end{equation}
where $\sum_{i\in X}K^{\dagger}_iK_i=I$ and $X$ is an index set of cardinality at most $d^2 $, with $d $ the dimension of ${\mathcal H} $. CPTP maps obviously leave $\mathcal{S}(\mathcal{H})$ invariant and hence induce a restricted map $T: \mathcal{S}(\mathcal{H})\longrightarrow\mathcal{S}(\mathcal{H})$ that we shall denote with the same symbol and use interchangeably. Note that, unlike $\mathcal{B}({\mathcal H}) $, the set $\mathcal{S}({\mathcal H}) $ is not a vector space, and hence it is only when use the map $T: \mathcal{B}(\mathcal{H})\longrightarrow\mathcal{B}(\mathcal{H})$ that we can talk about matrix expressions and eigenvalues for the operator $T$. It can be shown that any CPTP map $T: \mathcal{S}(\mathcal{H})\longrightarrow\mathcal{S}(\mathcal{H})$ is {\it non-expansive} in the trace norm, which means that after applying $T$ to two input states $\rho_1,\rho_2 \in \mathcal{S}(\mathcal{H})$, the distance between these density matrices is either contracted or remains equal:
\begin{equation}
\label{contraction condition}
    ||T(\rho_1)-T(\rho_2)||_{1}\leq||\rho_1-\rho_2||_1. 
\end{equation}
We recall that the {\it eigenvalues} of a CPTP map $T: \mathcal{B}(\mathcal{H})\rightarrow \mathcal{B}(\mathcal{H})$ are the complex numbers $\lambda$ that make the map $T-\lambda\, \text{id}$ non-invertible. We will denote this set by $\text{spec}(T)$. The {\it spectral radius} $\rho(T):=\max\{|\lambda|\ \mid \lambda\in\text{spec}(T) \}$ of any CPTP map is one as a consequence of \eqref{contraction condition} and of the fact that $1$ is always an eigenvalue. Additionally, the set $\text{spec}(T)$ is invariant under complex conjugation (see \cite{burgarth2013ergodic}).
On the other hand, since the finite-dimensional vector space $\mathcal{B}(\mathcal{H})$ is isomorphic to $\bb{C}^{d^2}$, we can represent $T: \mathcal{B}(\mathcal{H})\rightarrow \mathcal{B}(\mathcal{H})$ as a $d^2\times d^2$ matrix, which we write as $\widehat{T}$. This matrix representation can be obtained by fixing an orthonormal basis in $\mathcal{B}(\mathcal{H})$ for the Hilbert-Schmidt inner product, that is, $\{B_i\in\mathcal{B}(\mathcal{H})\}$ with $\tr(B_i^{\dagger}B_j)=\delta_{ij}$, and by setting $\widehat{T}_{ij}=\tr(B_i^{\dagger}T(B_j))$. 
The spectrum of the matrix $\widehat{T}$ is given by the roots of the characteristic polynomial $\det(\widehat{T}-\lambda I)=0$ and, since we are working with finite dimensions, it coincides with $\text{spec}(T)$, as well as with the set of complex values $\lambda$ that satisfy $T(X)=\lambda X$, for some non-trivial eigenvector $X\in \mathcal{B}(\mathcal{H})$  \cite{wolf2012quantum}.

We are particularly interested in the eigenvectors corresponding to the eigenvalue $\lambda =1$ and that we call {\it fixed points}, that is, they are elements $A\in\mathcal{B}(\mathcal{H})$ such that  $T(A)=A$. The set of fixed points of a CPTP map is always non-empty in the space of density matrices. This is a consequence of Schauder's Fixed-Point Theorem together with the continuity of $T$ and the compactness of $\mathcal{S}(\mathcal{H})$. We will work most of the time with quantum channels with single fixed points $\rho^*\in \mathcal{S}(\mathcal{H})$ in the space of density matrices. Such maps are called {\it ergodic}. Ergodicity of quantum channels is equivalent to the eigenspace of $T$ associated with the eigenvalue $\lambda=1  $ in $\mathcal{B}(\mathcal{H})$ having dimension one, and to being made out of the complex multiples of $\rho^*\in \mathcal{S}(\mathcal{H})$. In that case, we obviously have  (see Corollary 2 in \cite{burgarth2013ergodic}) that  $T(A)=A$ with $A\in \mathcal{B}(\mathcal{H})$  if and only if $A=\tr(A)\rho^*$. It is worth mentioning that the rest of the eigenvectors of an ergodic CPTP map must be traceless. Indeed, given an eigenvalue $\lambda\neq 1 $ of the CPTP map $T$ and $A$ a corresponding eigenvector, we find that $\tr(A)=\tr(T(A))=\lambda\tr(A)$, which implies that $ (\lambda -1)\tr(A)=0$, and hence that $ \tr(A)=0$. Therefore, the spectral set of a CPTP map can be decomposed as:
\begin{equation}
    \text{spec}(T)=\{1\}\cup \text{spec}(T|_{\mathcal{B}_0(\mathcal{H})}), 
\end{equation}
where $\mathcal{B}_0(\mathcal{H})\subset \mathcal{B}(\mathcal{H})$ is defined as the vector subspace of traceless operators, where the restriction $T|_{\mathcal{B}_0(\mathcal{H})}$ and its corresponding matrix representation are obviously well defined.

Repeated applications of an ergodic CPTP map do not necessarily converge to a fixed point (see \cite{burgarth2013ergodic} for an example). If convergence to a fixed point takes place, we say that the CPTP map is {\it mixing}. More specifically, a CPTP map is mixing if and only if its repeated applications converge in the trace norm, that is,
\begin{equation}\label{eq:mixing}
    \lim_{n\rightarrow \infty}||T^n(\rho)-\rho^*||_1=0, \quad \forall \rho\in  \mathcal{S}(\mathcal{H}).
\end{equation}
Equation~\eqref{eq:mixing} implies that the sequence $\{T^n(\rho)\}$ converges to $\rho^*$ with respect to the trace norm, but in infinite-dimensional situations this does not necessarily imply that convergence takes place with respect to other norms (see Definition 5.4.1 in \cite{horn2012matrix}). However, in the finite-dimensional case mixing becomes a topological property where $\lim_{n\rightarrow \infty} T^n(\rho)=\rho^*$ for any norm and for any $\rho\in  \mathcal{S}(\mathcal{H}) $ \cite{burgarth2007generalized}.

 Although all mixing maps are ergodic, the converse is not true in general. However, for continuous-time Markovian evolution, both are equivalent and such maps are called {\it relaxing} \cite{burgarth2013ergodic}. Another important consequence of mixing condition is the following: a CPTP map is mixing if and only if the fixed point $\rho^*$ is the only eigenvector with eigenvalue $|\lambda|=1$. Then, it is straightforward to show that 
\begin{equation}
\label{condition to be mixing}
\max\{|\lambda|\mid \lambda\in\text{spec}(T|_{\mathcal{B}_0(\mathcal{H})}) \}<1. 
\end{equation}

We say that a CPTP map is called {\it primitive} when it is mixing and its unique fixed point $\rho^*>0$ has full rank. 
The smallest natural number $n$ for which $T^n$ sends positive semidefinite matrices to positive definite matrices is called the {\it index of primitivity} of $T$ and is denoted by $\omega(T)$. With that notation, we say that  $T^{\omega(T)}$ is a strictly positive map. Bounds for this number are given by the quantum version of the Wietland inequality \cite{sanz2010quantum,rahaman2019new}. 

Finally, we say that a quantum channel is {\it strictly contractive} when
\begin{equation}
\label{contraction}
    ||T(\rho_1)-T(\rho_2)||_{1}\leq r||\rho_1-\rho_2||_1
\end{equation}
for all $\rho_1,\rho_2\in \mathcal{S}(\mathcal{H})$, where $0\leq r<1$. As it is customary in the quantum channels literature, we reserve the term strictly contractive for the trace norm, unless a different norm is explicitly specified. It can be shown that strictly contractive channels 
with a full-rank fixed point must be primitive. 
Indeed, the contractivity condition, together with Banach's Fixed Point Theorem (using that the convex closed subset of density matrices with the trace norm is a complete metric space) guarantees that the channel is mixing. A mixing channel with a strictly positive fixed point is a primitive channel. The converse holds when the primitive map becomes strictly positive, that is, it sends positive semidefinite matrices to positive definite ones: given a primitive channel $T$, $T^{\omega(T)}$ is strictly contractive  (see Theorem VI.3 in \cite{rahaman2019new}). Contractive maps can be also constructed by composing a strictly contractive channel with a general CPTP map.

An important conclusion that can be drawn from all these considerations is that strictly contractive channels, which will be relevant along this work, can be constructed by either using a map like $T^{\omega(T)}$, where $T$ is a primitive channel, or by composing a strictly contractive channel with any other CPTP map, as it has been done in some examples in QRC \cite{kubota2022quantum,suzuki2022natural}. Examples of mixing/relaxing CPTP maps with full-rank single fixed points can be found in the Lindblad-Gorini-Kossakowski-Sudarshan equation (see \cite{nigro2019uniqueness} for a summary of the necessary and sufficient conditions). As these maps belong to the quantum Markov semigroup, iterative applications of the channels yield $T_{\Delta \tau }^{\omega(T)}=T_{\omega(T)\Delta \tau }$, where $\Delta \tau $ represents the time of action of the map $T_{\Delta \tau }$. Therefore, taking $\Delta \tau '\geq\omega(T)\Delta \tau $ we obtain a strictly contractive map.

\subsection{RC definitions}
\label{sec:rc_definitions}

We now define quatum reservoir computing (QRC) systems in the density matrix formalism. Classical RC definitions and mathematical details can be found in, for instance, \cite{grigoryeva2018echo}. QRC maps are determined by two equations, namely, the {\it state-space} and the {\it readout} or {\it observation} equations.  The state equation is given by a family of CPTP maps $T:\mathcal{B}(\mathcal{H})\times \bb{R}^n\rightarrow \mathcal{B}(\mathcal{H})$, with $n\in \bb{N}$ being the number of input features, which are taken to be real values (classical inputs). The maps $T$ and $h $ will be, most of the time, tacitly assumed to be continuous.  The output is obtained from the readout map $h:\mathcal{B}(\mathcal{H})\rightarrow \bb{R}^m$, with $m\in \bb{N}$, which maps  operators in $\mathcal{B}(\mathcal{H})$ to the Euclidean space $\bb{R}^m$. 
Inputs are typically bi-infinite discrete-time sequences of the form ${\bf z}=(\dots,{\bf z}_{-1},{\bf z}_0,{\bf z}_1,\dots)\in (\bb{R}^n)^{\bb{Z}}$, and outputs $\textbf{y}\in(\bb{R}^m)^{\bb{Z}}$ have the same structure. A QRC system is hence determined by the state-space transformations:
\begin{equation} \label{eq:QRC1}
    \begin{cases}
    &A_t=T(A_{t-1},{\bf z}_t),\\
    &\textbf{y}_t=h(A_t),
    \end{cases}
\end{equation} 
where $t\in\bb{Z}$ denotes the time index. Analogously, one can define the same setting for semi-infinite discrete-time sequences: $(\bb{R}^n)^{\bb{Z}_-}=\{{\bf z}=(\dots,{\bf z}_{-1},{\bf z}_0)\ |\ {\bf z}_i\in \bb{R}^n, i \in \bb{Z}_- \}$ for left-infinite sequences and $(\bb{R}^n)^{\bb{Z}_+}=\{{\bf z}=({\bf z}_0,{\bf z}_1,\dots)\ |\ {\bf z}_i\in \bb{R}^n, i \in \bb{Z}_+ \}$ for right-infinite sequences. Similar definitions apply to $(D_n)^{\bb{Z}}$, $(D_n)^{\bb{Z}_-}$, and $(D_n)^{\bb{Z}_+}$ with elements in the subset $D_n\subset \bb{R}^n$. We can also construct sequence spaces $(\mathcal{B}(\mathcal{H}))^{\bb{Z}} $ for the space of bounded (trace-class) operators:
\begin{multline*}
    (\mathcal{B}(\mathcal{H}))^{\bb{Z}}= \{\textbf{A}=(\dots,A_{-1},A_0,A_1,\dots)\\
\mid A_i\in\mathcal{B}(\mathcal{H}),\ i\in\bb{Z}\}. 
\end{multline*}
Analogous definitions for $(\mathcal{B}(\mathcal{H}))^{\bb{Z}_-}$,$(\mathcal{B}(\mathcal{H}))^{\bb{Z}_+}$, $(\mathcal{S}(\mathcal{H}))^{\bb{Z}}$, $(\mathcal{S}(\mathcal{H}))^{\bb{Z}_-}$, and $(\mathcal{S}(\mathcal{H}))^{\bb{Z}_+}$  follow immediately.

A natural way to construct CPTP state-space transformations is to insert the input dependence using the Kraus decomposition that we introduced in \eqref{kraus decomposition}, that is, 
\begin{equation}
\label{Kraus with inputs}
T(A, {\bf z})=\sum_{i\in X}K_i({\bf z})AK^{\dagger}_i({\bf z}),
\end{equation}
and for any input ${\bf z} \in \bb{R}^n$, the matrices $\{K_i({\bf z}) \}_{i \in X}$ satisfy that $\sum_{i\in X}K^{\dagger}_i({\bf z})K_i({\bf z})=I$. The by-design CPTP character of the map $T:\mathcal{B}(\mathcal{H})\times \bb{R}^n\rightarrow \mathcal{B}(\mathcal{H})$ in \eqref{Kraus with inputs} implies that it naturally restricts to a state equation $T:\mathcal{S}(\mathcal{H})\times \bb{R}^n\rightarrow \mathcal{S}(\mathcal{H})$ with density matrices as state space, that we shall use interchangeably in the sequel and denote using the same symbol.

\medskip

\noindent \textit{The Echo State Property (ESP)}. Consider the QRC system defined in~\eqref{eq:QRC1} or its analog for the subsets $\mathcal{S}(\mathcal{H})\subset \mathcal{B}(\mathcal{H})$ and $D_n\subset \bb{R}^n$, that is,  $T:\mathcal{S}(\mathcal{H})\times D_n\rightarrow \mathcal{S}(\mathcal{H})$. Given an input sequence ${\bf z} \in (D_n)^{\bb{Z}}$, we say that $\boldsymbol{\rho} \in  (\mathcal{S}(\mathcal{H}))^{\bb{Z}} $ is a {\it solution} of~\eqref{eq:QRC1} for the input ${\bf z} $ if the  components of the sequences ${\bf z} $ and $\boldsymbol{\rho}$ satisfy the first relation in~\eqref{eq:QRC1} for any $t \in \Bbb Z$.  We say that the QRC system has the {\it echo state property} (ESP) when it has a unique solution for each input ${\bf z} \in (D_n)^{\bb{Z}}$. More explicitly, for each ${\bf z}\in (D_n)^{\bb{Z}}$, there exists a unique sequence $\boldsymbol{\rho}\in (\mathcal{S}(\mathcal{H}))^{\bb{Z}}$ such that 
\begin{equation}
\label{eq:QRC2}
    \rho_t=T(\rho_{t-1},{\bf z}_t), \ \text{for all} \ t\in\bb{Z}.
\end{equation}

\medskip

\noindent {\it Filters and functionals}. Let $\mathcal{S}(\mathcal{H})\subset \mathcal{B}(\mathcal{H})$ be the space of density matrices and let $D_n\subset \bb{R}^n$ be a subset in the input space. A map of the type $U: (D_n)^{\bb{Z}} \rightarrow (\mathcal{S}(\mathcal{H}))^{\bb{Z}}$ is called a {\it filter} associated to the QRC system~\eqref{eq:QRC1} when it satisfies that 
\begin{equation*}
U({\bf z}) _t=T \left(U({\bf z})_{t-1}, {\bf z} _t\right),\ \mbox{for all ${\bf z} \in (D_n)^{\bb{Z}}$ and $t \in \Bbb Z $.}
\end{equation*}
Filters induce what we call {\it functionals} $H: (D_n)^{\bb{Z}} \rightarrow \mathcal{S}(\mathcal{H})$ via the relation $H({\bf z})=U({\bf z}) _0 $. It is clear that a uniquely determined filter can be associated with a QRC system that satisfies the ESP. The filter maps, in that case, any input sequence to the unique solution of the QRC system associated with it. A filter is called {\it causal} if it only produces outputs that depend on present and past inputs. More formally, causality means that for any two inputs ${\bf z},\textbf{v}\in (D_n)^{\bb{Z}}$ that satisfy ${\bf z}_{\tau}=\textbf{v}_{\tau}$ for any $\tau\leq t$, for a given $t\in\bb{Z}$, we have $U({\bf z})_t=U(\textbf{v})_t$. The filter $U$ is called {\it time-invariant} if there is no explicit time dependence on the system that determines it, that is, it commutes with the {\it time delay operator} defined as $\mathcal{T}_{\tau}({\bf z})_t:={\bf z}_{t-\tau}$. Filters associated to QRC systems of the type~\eqref{eq:QRC1} are always causal and time-invariant (Proposition 2.1 in \cite{grigoryeva2018echo}). As noted in previous works \cite{boyd1985fading, grigoryeva2018echo, grigoryeva2018universal}, there is a bijection between causal and time-invariant filters and functionals on $(D_n)^{\bb{Z}_-}$.  Then, we can restrict our work to causal and time-invariant filters with target and domain in spaces of left-infinite sequences. 

\medskip

\noindent \textit{The Fading Memory Property (FMP)}. 
Infinite product spaces can be endowed with Banach space structures associated to the supremum norm and weighted norms. The supremum norm for input sequences is defined as $||{\bf z}||_{\infty}:=\sup_{t\in \bb{Z}}\{||{\bf z}_t||\}$, and for operator sequences as $||\textbf{A}||_{\infty}:=\sup_{t\in \bb{Z_-}}\{||A_t||\}$, where $||\cdot||$ represents a given vector and matrix norm, respectively. The symbols $l^{\infty}(\bb{R}^n)$, $l_{\pm}^{\infty}(\bb{R}^n)$, $l^{\infty}(\mathcal{B}(\mathcal{H}))$ and $l_{\pm}^{\infty}(\mathcal{B}(\mathcal{H}))$ denote the Banach spaces formed by the elements in the corresponding infinite product spaces with a finite supremum norm.

We now define the weighted norm. Let $w:\bb{N}\rightarrow (0,1]$ be a decreasing sequence with zero limit and $w_0=1$. The weighted norm $||\cdot||_w$ on $(\bb{R})^{\bb{Z_-}}$ is defined as 
\begin{equation}
    ||{\bf z}||_w:=\sup_{t\in \bb{Z}_-}\{w_{-t}||{\bf z}_t||\},
\end{equation}
and the space 
\begin{equation}
    l^{w}_-(\bb{R}^n)=\{{\bf z}\in (\bb{R}^n)^{\bb{Z_-}}|\ ||{\bf z}||_w<\infty\},
\end{equation}
with weighted norm $||\cdot||_w$ forms a Banach space (see Appendix A.2 in \cite{grigoryeva2018echo}).
In the same vein, we can define 
\begin{equation}
\begin{split}
    &||\textbf{A}||_w:=\sup_{t\in \bb{Z}_-}\{w_{-t}||A||\}, \\
    &l^{w}_-(\mathcal{B}(\mathcal{H}))=\{\textbf{A}\in (\mathcal{B}(\mathcal{H}))^{\bb{Z_-}}|\ ||\textbf{A}||_w<\infty\}.
    \end{split}
\end{equation}
It can be shown that $l^{w}_-(\mathcal{B}(\mathcal{H}))$ is a Banach space as well. 

We now turn to the space of density matrices $\mathcal{S}(\mathcal{H})  $ and recall that since for positive semidefinite matrices the trace operator is submultiplicative (see Exercise 7.2.P26 in \cite{horn2012matrix}) we can conclude that $ \operatorname{tr} (\rho ^2)\leq 1 $ for all $\rho \in \mathcal{S}(\mathcal{H})  $. This observation implies that the elements in $\mathcal{S}(\mathcal{H})  $ have  Frobenius norms bounded by one. Since we are in finite dimensions, this statement holds for any other matrix norm (eventually with a bounding constant different from one) and allows us to conclude that   
\begin{equation}
\label{inclusions spaces w}
(\mathcal{S}(\mathcal{H}))^{\bb{Z}_-}\subset  l^{\infty}_-(\mathcal{B}(\mathcal{H}))\subset l^{w}_-(\mathcal{B}(\mathcal{H})),
\end{equation}
for any weighting sequence $w$. The previous boundedness consideration, together with the closedness of $\mathcal{S}(\mathcal{H}) $ in $\mathcal{B}(\mathcal{H}) $ implies that $\mathcal{S}(\mathcal{H}) $ is necessary compact. An important consequence of this fact is that the relative topology induced by the $l^{w}_-(\mathcal{B}(\mathcal{H})) $ on  $(\mathcal{S}(\mathcal{H}))^{\bb{Z}_-} $ coincides with the product topology (see Corollary 2.7 in  \cite{grigoryeva2018echo}).

Take now a subset $D_n\subset \bb{R}^n$ such that  $(D_n)^{\bb{Z}_-}\subset l^{w}_-(\bb{R}^n)$ and consider a QRC system $T:\mathcal{S}(\mathcal{H})\times D_n\rightarrow \mathcal{S}(\mathcal{H})$ that has the ESP. We say that $T$ has the {\it fading memory property} (FMP) when the corresponding functional $H:(D_n)^{\bb{Z}_-} \rightarrow\mathcal{S}(\mathcal{H})$ is a continuous map between the metric spaces $((D_n)^{\bb{Z}_-},||\cdot||_w)$ and $((\mathcal{S}(\mathcal{H}))^{\bb{Z}_-},||\cdot||_w)$, for some weighting sequence $w$. If  $D_n$ is compact, once $H$ is continuous for a given weighting sequence $w$, then it is continuous for all  weighting sequences (see \cite[Theorem 2.6]{grigoryeva2018echo}).

The compactness of $\mathcal{S}(\mathcal{H}) $ implies that we can apply with straightforward modifications Theorem 3.1 in \cite{grigoryeva2018echo} to prove the following statement. 

\begin{proposition}
\label{prop:contractivity implies ESP and FMP}
    Let $D_n\subset \bb{R}^n$ be a compact subset of $\mathbb{R}^n $. Let $T:\mathcal{S}(\mathcal{H})\times D_n\rightarrow \mathcal{S}(\mathcal{H})$ be a continuous QRC system such that  the CPTP maps $T(\cdot,{\bf z}):\mathcal{S}(\mathcal{H})\rightarrow \mathcal{S}(\mathcal{H})$ are strictly contractive for all ${\bf z} \in D_n$ as in \eqref{contraction} with a common contraction constant $0\leq r<1$ associated to some norm in $\mathcal{B}(\mathcal{H})$ (not necessarily $\left\|\cdot \right\|_1 $).  Then, the QRC system induced by $T$ has the ESP and the FMP.
\end{proposition}

\section{Results}
\label{sec:results}

Finding the expression for the QRC filter of a system that has the ESP is not straightforward in the language of density matrices since the dependence on the initial condition has to be addressed. More explicitly, given a general CPTP map expressed using its Kraus decomposition, the relation \eqref{Kraus with inputs} can be iterated $n$-time steps into the past in order to obtain a quantum state $\rho _t ^n(\rho^0_{t-n})\in \mathcal{S}({\mathcal H})$ at time $t$ out of an initial condition $\rho^0_{t-n}\in \mathcal{S}({\mathcal H})$ specified at time $t-n $ via the formula:
\begin{equation}
\footnotesize
\label{eq:kraus_t}
\begin{split}
&\!\!\!\!\rho _t ^n(\rho^0_{t-n})=\\
&\sum_{i_0,\ldots,i_{n-1}\in X}\left(\overleftarrow{\prod}^{n-1}_{l=0}K_{i_l}({\bf z}_{t-l})\right)\rho^0_{t-n}\left(\overrightarrow{\prod}^{n-1}_{l=0}K^{\dagger}_{i_l}({\bf z}_{t-l})\right),  
\end{split}
\end{equation}
where $\overleftarrow{\prod}^{n-1}_{l=0}K_{i_l}({\bf z}_{t-l})=K_{i_0}({\bf z} _t)\cdots K_{i_{n-1}}({\bf z}_{t-n-1})$ and $\overrightarrow{\prod}^{n-1}_{l=0}K^{\dagger}_{i_l}({\bf z}_{t-l})=K^{\dagger}_{i_{n-1}}({\bf z}_{t-n-1})\cdots K^{\dagger}_{i_0}({\bf z} _t)$. If the QRC system has the ESP and hence a filter $U: (D_n)^{\mathbb{Z}} \longrightarrow \left(\mathcal{S}({\mathcal H})\right)^{\mathbb{Z}} $ can be associated to it, it necessarily has to satisfy
\begin{equation}
\label{eq:filter_rho}
    U({\bf z})_t=\lim_{n\rightarrow \infty}\rho _t ^n(\rho^0_{t-n}),
\end{equation}
and this value has to be independent of the initial conditions $\rho^0_{t-n} $. This fact has been explicitly shown in \cite{chen2019learning} (and in Chapter 2 of \cite{chen2022nonlinear} in more detail) in the case of strictly contractive CPTP maps with respect to the {\it operator norm} $\vertiii{\cdot }_p$ associated to Schatten norms. We recall that given $T: \mathcal{B}({\mathcal H}) \rightarrow \mathcal{B}({\mathcal H}) $, we define
\begin{equation}
\vertiii{T}_p:=\sup_{||A||_p=1}\{||T(A)||_p\mid A\in \mathcal{B}({\mathcal H})\}, 
\end{equation}
for some $p\in [1, \infty)$ and $\left\|\cdot \right\|_p $ the $p$-Schatten norm.
Using this notation, the contractivity condition on the CPTP map $T$ is stated by requiring that $\vertiii{T|_{\mathcal{B}_0(\mathcal{H})}}_p<1$, for some $p\in [1, \infty)$.  It is obvious that this condition ensures that the hypotheses of Proposition \ref{prop:contractivity implies ESP and FMP} are satisfied, which in turn implies that $T$ has the ESP and the FMP (notice that given $\rho _1, \rho _2 \in \mathcal{S}({\mathcal H})$ and ${\bf z} \in D _n $ arbitrary, the difference $T(\rho _1, {\bf z})-T(\rho _2, {\bf z}) \in \mathcal{B}_0(\mathcal{H})$ and hence the hypothesis  $\vertiii{T|_{\mathcal{B}_0(\mathcal{H})}}_p<1$ implies that $\left\|T(\rho _1, {\bf z})-T(\rho _2, {\bf z})\right\| _p< \left\|\rho_1- \rho _2\right\|_p $). Since we are in a finite-dimensional context, it suffices to impose contractivity for just one norm $\vertiii{\cdot }_p$, $p\in [1, \infty)$, in order to ensure that the limit in \eqref{eq:filter_rho} exists and that its value is independent of the initial condition.

The expression~\eqref{eq:kraus_t} shows that it is not possible to write down a closed-form expression for the filter of a QRC system with the ESP when using the density matrix representation, in which the initial dependence condition has been eliminated. This feature is ultimately due to the linearity of the setup. Already in the classical framework (see \cite{grigoryeva2018universal}), it has been shown that this can be avoided by working with affine instead of purely linear systems. In the quantum context, it can also be observed (see \cite{chen2020temporal,molteni2022optimization}) that the non-homogeneous state-space system given by $\rho_t=(1-\epsilon)T(\rho_{t-1},{\bf z}_t)+\epsilon\sigma$, where $T(\rho_{t-1},{\bf z}_t)$ is a CPTP map, $\sigma$ an arbitrary density matrix, and $0<\epsilon<1$, defines a unique filter $U({\bf z})_t=\epsilon \sigma +\epsilon\sum^{\infty}_{j=1}(1-\epsilon)^j\overleftarrow{\prod}^{j-1}_{k=0}T(\sigma,{\bf z}_{t-k})$ in which the dependence on initial conditions has disappeared. 

In the following subsections, we shall circumvent this problem by showing that certain matrix representations of finite-dimensional QRC systems on density matrices have a built-in non-homogeneous affine structure that makes them into {\it non-homogeneous state-affine systems} (SAS) of the type introduced in \cite{grigoryeva2018universal}. More specifically, we shall be working with the {\it Bloch vector representation} of quantum finite dimensional systems \cite{kimura2003bloch,byrd2003characterization} associated to a given {\it Gell-Mann basis}. This idea is not new in QRC and it can already be seen in the seminal work \cite{fujii2017harnessing} or, more recently, in the Methods section of \cite{kubota2022quantum}. In the paragraphs that follow, we shall explore in depth this representation, mostly in connection with the available literature on SAS systems (Section \ref{Non-homogenous state affine system representation}), which will allow us later on in Section \ref{Some constrains on CPTP maps for QRC} to identify various design constraints on quantum channels.

\subsection{Matrix representation of quantum channels}
  We will start by introducing the notation necessary for the matrix representation of quantum channels. In the next section, we shall focus on a specific choice of basis adapted to density matrices. Let $\{B_i\}_{i \in \left\{1, \ldots, d ^2\right\}}$ be an orthonormal basis for the vector space $\mathcal{B}(\mathcal{H})$, when endowed with the Hilbert-Schmidt inner product, that is, $\tr(B^{\dagger}_iB_j)=\delta_{ij}$. Using any such basis we can represent any operator  $A \in \mathcal{B}(\mathcal{H})$ as $A=\sum_{i=1}^{d^2} a_iB_i$, with $a_i=\tr(B_i^{\dagger}A)$. Analogously, we can express any linear map $T: \mathcal{B}(\mathcal{H})\rightarrow \mathcal{B}(\mathcal{H})$ as
\begin{equation}
\label{eq:Wmatrix}
    T(A)
=\sum^{d^2}_{i,j=1}\widehat{T}_{ij}a_jB_i, \ \mbox{where $\widehat{T}_{ij}=\tr(B^{\dagger}_iT(B_j))$.}
\end{equation}
This observation implies that QRC systems $T:\mathcal{S}(\mathcal{H})\times D_n\rightarrow \mathcal{S}(\mathcal{H})$ admit an equivalent representation as a system 
$\widehat{T}: V \times D_n\rightarrow V$, where $V\subset \bb{R}^{d^2}$ is the subset of real Euclidean space that contains the coordinate representations of the elements in $\mathcal{S}(\mathcal{H}) $ using the basis $\{B_i\}_{i \in \left\{1, \ldots, d ^2\right\}}$. 

This statement can be formalized using the language of {\it system morphisms} (see \cite{RC15, RC16} for the standard definitions and elementary facts). Consider the state-space systems determined by the triples $({\cal X} _i, F _i, h _i)$, $i \in \left\{1,2\right\}$, with $F _i: {\cal X} _i\times {\cal Z}\longrightarrow {\cal X}_i$ and $h _i:{\cal X}_i \longrightarrow {\cal Y} $.
A map $f: {\cal X} _1 \longrightarrow {\cal X}_2$ is a {\it  morphism} between the systems $({\cal X} _1,F_1, h_1)$ and $({\cal X} _2, F_2, h_2)$ whenever it satisfies the following two properties:
\begin{enumerate}[(i)]
\item {\it System equivariance:} $f(F_1({\bf x}_1, {\bf z})) = F_2(f({\bf x}_1), {\bf z})$, for all ${\bf x}_1 \in {\cal X}_1$ and ${\bf z} \in {\cal Z}$.
\item {\it Readout invariance:} $h_1({\bf x}_1) = h_2(f({\bf x}_1))$, for all ${\bf x}_1 \in {\cal X} _1$.
\end{enumerate}
When the map $f$ has an inverse $f ^{-1} $  and this inverse is also a morphism between the systems determined by  $({\cal X} _2, F_2, h_2)$ and $({\cal X} _1, F_1, h_1)$, we say that $f$ is a {\it system isomorphism} and the systems $({\cal X} _1, F_1, h_1)$ and $({\cal X} _2, F_2, h_2)$ are {\it isomorphic}. We note that given a system $F_1: {\cal X}_1 \times {\cal Z} \longrightarrow {\cal X}_1, h_1: {\cal X}_1 \longrightarrow {\cal Y}$ and a bijection $f:{\cal X}_1 \longrightarrow {\cal X}_2$, the map $f$ is a system isomorphism with respect to the system $F_2: {\cal X}_2 \times {\cal Z} \longrightarrow {\cal X}_2, h_2: {\cal X}_2 \longrightarrow {\cal Y}$ defined by 
\begin{align}
F_2({\bf x}_2, {\bf z}) &:= f(F_1(f^{-1}({\bf x}_2), {\bf z})), \label{isomorphic state map}\\
h_2({\bf x}_2) &:= h_1(f^{-1}({\bf x}_2)).\label{isomorphic readout map}
\end{align}
for all $ {\bf x}_2 \in {\cal X}_2$, ${\bf z} \in {\cal Z}$.

Consider now the QRC system given by the quantum channel $T:\mathcal{S}(\mathcal{H})\times D _n\rightarrow \mathcal{S}(\mathcal{H})$ and the readout map $h:\mathcal{B}(\mathcal{H})\rightarrow \bb{R}^m$, $m\in \bb{N}$. Using the orthonormal basis $\mathcal{B}=\{B_i\}_{i \in \left\{1, \ldots, d ^2\right\}}$ and the discussion above define the map
\begin{equation}
\label{system isom linear}
\begin{array}{cccc}
G_{\mathcal{B}}: &\mathbb{C}^{d^2} &\longrightarrow & \mathcal{B}({\mathcal H})\\
	&\mathbf{a} &\longmapsto &\sum_{i=1}^{d^2} a _i B _i.
\end{array}
\end{equation}
This map is a linear homeomorphism. Define $V=G_{\mathcal{B}} ^{-1}(\mathcal{S}(\mathcal{H}))\subset \bb{R}^{d^2}$ as well as the map (that we denote with the same symbol) $G_{\mathcal{B}}:V \longrightarrow \mathcal{S}(\mathcal{H}) $  that we obtain by restriction of the domain and codomain in \eqref{system isom linear}. This restricted map is also a homeomorphism when $V$ and $\mathcal{S}(\mathcal{H}) $ are endowed with their relative topologies \cite[Theorem 18.2]{munkres2000topology}. With all these ingredients, it is straightforward to verify that the QRC system $(\mathcal{S}(\mathcal{H}), T, h) $ is system isomorphic to $(V, \widehat{T}, \widehat{h}) $ with $\widehat{T}:V \times D_n \longrightarrow V $  and $\widehat{h}: V \longrightarrow \mathbb{R}^m $ given by
\begin{align}
\widehat{T}(\mathbf{a}, {\bf z})&:= G_{\mathcal{B}}^{-1} \left(T \left(G_{\mathcal{B}}(\mathbf{a}), {\bf z}\right)\right), \label{isomorphic state map qrc}\\
\widehat{h}(\mathbf{a})&:= h(G_{\mathcal{B}}(\mathbf{a})),\label{isomorphic readout map qrc}
\end{align}
and that the isomorphism is given by the map $G_{\mathcal{B}}:V \longrightarrow \mathcal{S}(\mathcal{H}) $. The procedure that we just spelled out can be reproduced for any other (orthonormal) basis $\mathcal{B} '$ of $\mathcal{B}({\mathcal H}) $, in which case we would obtain another system $(V', \widehat{T}', \widehat{h}') $ which is obviously isomorphic to both $(V, \widehat{T}, \widehat{h}) $ and $(\mathcal{S}(\mathcal{H}), T, h) $.

The system isomorphisms that we just defined and the compactness of the state spaces where they are defined allow us to establish important connections between the filters that they define. The following proposition describes those connections in detail.

\begin{proposition} 
\label{prop:isomorphism} 
Let $T:\mathcal{S}(\mathcal{H})\times D _n\rightarrow \mathcal{S}(\mathcal{H})$ be a quantum channel, $h: \mathcal{S}(\mathcal{H}) \longrightarrow \mathbb{R} ^m $ a readout, and let $\mathcal{B}=\{B_i\}_{i \in \left\{1, \ldots, d ^2\right\}}$ be an orthonormal basis for $\mathcal{B}(\mathcal{H})$. Let $\widehat{T}: V \times D_n \longrightarrow V $ be the isomorphic system defined in \eqref{system isom linear} and $\widehat{h}: V \longrightarrow \mathbb{R}^m $ the corresponding readout defined in \eqref{isomorphic readout map qrc}. Then:
\begin{enumerate}[(i)]
\item Given an input ${\bf z} \in ({\Bbb R}^n)^{\mathbb{Z}} $, a sequence $\boldsymbol{\rho} \in \left(\mathcal{S}(\mathcal{H})\right)^{\mathbb{Z}}  $ is a solution for that input for the system determined by $T$ if and only if the sequence ${\cal G}_{\mathcal{B}}^{-1} \left(\boldsymbol{\rho}\right) \in \left(V\right) ^{\mathbb{Z}} $ is a solution for the system associated to $\widehat{T} $. The symbol ${\cal G}_{\mathcal{B}}=\prod _{\Bbb Z} G _{\mathcal{B}}: (V)^{\mathbb{Z}}\longrightarrow \left({\cal S}( {\mathcal H})\right)^{\mathbb{Z}}$ stands for the product map.
\item  $T$ has the ESP if and only if $\widehat{T} $ has the ESP. In that case, the filters $U _T $ and $U_{\widehat{T}}$ (respectively,  $U _T^h$ and  $U_{\widehat{T}}^{\widehat{h}} $) determined by $T$  and $\widehat{T} $ (respectively, by $(T,h)$  and $(\widehat{T}, \widehat{h}) $) satisfy that $U _T={\cal G}_{\mathcal{B}} \circ U_{\widehat{T}} $ (respectively, $U _T^h= U_{\widehat{T}}^{\widehat{h}} $).
\item $U _T  $  has the FMP if and only if $U_{\widehat{T}} $ has the FMP.
\item Let $V'$ be a set homeomorphic to $V$. The relations  \eqref{isomorphic state map}-\eqref{isomorphic readout map} determine in that case an isomorphic system   $\widehat{T}':V ' \times D _n \longrightarrow V '$, $\widehat{h} ': V '  \rightarrow \mathbb{R}^m  $ that has the ESP and the FMP if and only if $\widehat{T} $ and $\widehat{h} $ have that property. 
\end{enumerate}
\end{proposition}

\begin{proof}
Parts \textit{(i)}, {\it (ii)}, and {\it (iv)} are a straightforward consequence of Proposition 2.2 in \cite{grigoryeva2018echo}. As to part \textit{(iii)}, given that by {\it (ii)}  $U _T={\cal G}_{\mathcal{B}} \circ U_{\widehat{T}} $, it suffices to prove that the bijection ${\cal G}_{\mathcal{B}}=\prod _{\Bbb Z} G _{\mathcal{B}}: (V)^{\mathbb{Z}}\longrightarrow \left({\cal S}( {\mathcal H})\right)^{\mathbb{Z}}$ is a homeomorphism when the domain and the target are endowed with a weighted norm. In order to prove that, recall first that, using the observation right under \eqref{inclusions spaces w}, the weighted norms in $(V)^{\mathbb{Z}} $ and $  \left({\cal S}( {\mathcal H})\right)^{\mathbb{Z}} $ induce the product topology due to the compactness of $V $ and $\mathcal{S}(\mathcal{H})$. This immediately implies that  (Theorem 19.6 in \cite{munkres2000topology}) ${\cal G}_{\mathcal{B}}=\prod _{\Bbb Z} G _{\mathcal{B}} $ is continuous due to the continuity of $G _{\mathcal{B}} $. The same argument can be immediately applied to the inverse map ${\cal G}_{\mathcal{B}} ^{-1} $, which proves the statement.
\end{proof}

\subsection{Non-homogenous state affine system representation}
\label{Non-homogenous state affine system representation} 

The goal of this section is to choose a specific basis in $\mathcal{B}({\mathcal H})$ for the representation of the quantum channel $T:\mathcal{B}(\mathcal{H})\times D _n\rightarrow \mathcal{B}(\mathcal{H})$ (equivalently, $T:\mathcal{S}(\mathcal{H})\times D _n\rightarrow \mathcal{S}(\mathcal{H})$) in which the associated system isomorphic representation $\widehat{T}: \mathbb{C}^{d^2} \times D_n \longrightarrow \mathbb{C}^{d^2}$ (equivalently, $\widehat{T}: V \times D_n \longrightarrow V$) has a non-homogeneous state-affine form of the type introduced in \cite{grigoryeva2018universal}.

More specifically, we choose a {\it generalized Gell-Mann basis} \cite{kimura2003bloch,byrd2003characterization,siewert2022orthogonal}. This is an orthonormal basis made of Hermitian operators in which, by convention, its first element is the normalized identity, namely, $B_1=I/\sqrt{d}$. The remaining $(d^2-1)$ traceless Hermitian operators are the generators 
\begin{equation}
\label{zero trace gens}
\mathcal{B}_0=\left\{\sigma_i\right\}_{i \in \left\{1, \ldots, d^2-1\right\}}
\end{equation}
of the {\it fundamental representation}   of the Lie algebra $\mathfrak{su}(d)$ of SU($d$). The case  $d=2$ corresponds to the case of one qubit, and the Gell-Mann basis is made of the standard Pauli matrices. The orthonormality of the Gell-Mann basis is guaranteed by the product property of the fundamental representation of  $\mathfrak{su} (d) $, namely, $\sigma_a\sigma_b=\delta_{ab}I/(2d)+\sum^{d^2-1}_cf_c\sigma_c$, where $\sigma_a, \sigma_b $ are two elements of the Gell-Mann basis for $\mathfrak{su} (d) $,  and $f_c$ are complex coefficients. The resulting Gell-Mann basis $\mathcal{B}=\left\{B_i\right\}_{i \in \left\{1, \ldots, d^2\right\}}$  of ${\cal B}({\mathcal H}) $ is hence given by $B _1=I/\sqrt{d}$ and $B _i= \sigma_{i-1} $, $1<i\leq d ^2 $. Note that the subset $\mathcal{B} _0=\left\{B_i\right\}_{i \in \left\{2, \ldots, d^2\right\}}$ is a basis for the vector subspace ${\cal B} _0({\mathcal H}) \subset {\cal B}({\mathcal H}) $ of codimension $1$ made of traceless operators.

If our system is made of $N$ $d$-dimensional systems ({\it qudits}), we can extend this basis to $\mathcal{B} \left({\mathcal H}^{\otimes ^N}\right) $ by tensorization.  The orthonormality of the tensorized basis with respect to the Hilbert-Schmidt inner product is preserved since $\tr(A\otimes B)=\tr(A)\tr(B)$ for any two $A,B \in \mathcal{B}({\mathcal H}) $. 

We now go back to the system constituted by one qudit and spell out the matrix expression $\widehat{T} :\mathbb{C}^{d^2} \times D _n \longrightarrow \mathbb{C}^{d^2} $ of a CPTP map $T: \mathcal{B}({\mathcal H})\times D_n\longrightarrow \mathcal{B}({\mathcal H})$ in the basis $\mathcal{B}  $  that we just introduced, by using the prescription introduced in \eqref{eq:Wmatrix}. 
We first note that the choice $B_1=I/\sqrt{d}$ and the trace-preserving character of $T(\cdot , {\bf z})$ for any ${\bf z} \in D_n $, imply that  $\widehat{T}(\cdot , {\bf z})_{11}  =1$. Analogously, $\widehat{T}(\cdot , {\bf z})_{1j}=\tr(T(B_j, {\bf z}))/\sqrt{d}=0$ for $1<j\leq d^2$, since $T$ is trace-preserving. This implies that the matrix $\widehat{T}(\cdot , {\bf z})$ can be written as
\begin{equation}
\label{matrix form t hat}
    \widehat{T}(\cdot , {\bf z})=\left(\begin{matrix}
        1 & {\bf 0}_{d^2-1}  \\
        q({\bf z}) & p({\bf z})  
    \end{matrix}\right)
\end{equation}
where $p({\bf z}) $ is the square matrix of dimension $d^2-1 $ with complex entries 
\begin{multline}\label{eq:p_ij}
p({\bf z})_{ij}:=\left(\widehat{T}(\cdot  , {\bf z})|_{G_{\mathcal{B}}^{-1} \left(\mathcal{B}_0(\mathcal{H})\right)}\right)_{ij}\\
=
\left(\widehat{T}(\cdot  , {\bf z})|_{G _{\mathcal{B}} ^{-1}\left(\operatorname{span} \left\{\mathcal{B} _0\right\}\right)}\right)_{ij}
=\tr(B_iT(B_j, {\bf z})),
\end{multline}
$1<i,j\leq d^2$, and $q({\bf z})\in \mathbb{C}^{d^2-1} $ is given by $q({\bf z})_{i}=\tr(B_iT(I))/\sqrt{d}$, $1<i\leq d^2$. The symbol ${\bf 0}_{d^2-1}$ denotes a zero-row of length $d^2-1$.

Now, given that any element $\rho \in \mathcal{S}({\mathcal H}) $ has coordinates in the Gell-Mann basis of the form $\left(1/\sqrt{d}, \mathbf{x}^{\top} \right)^{\top}\in \mathbb{C}^{d^2}   $ with $\mathbf{x} \in \mathbb{C}^{d^2-1} $ called the {\it Bloch vector}, the matrix form \eqref{matrix form t hat} implies that 
\begin{equation}
\label{matrix form t hat with vector}
    \widehat{T}(\cdot , {\bf z})\left(\begin{matrix}
        1/\sqrt{d}  \\
        \mathbf{x}  
    \end{matrix}\right)=
\left(\begin{matrix}
        1/\sqrt{d}  \\
        p({\bf z})\textbf{x}+q({\bf z})  
    \end{matrix}\right).
\end{equation}
This expression implies that  $\widehat{T} :V \times D _n \longrightarrow V $ admits a system-isomorphic representation $\widehat{T}_0 :V_0 \times D _n \longrightarrow V_0 $ on the set $V _0= \left\{\mathbf{x} \in \mathbb{C}^{d^2-1}\mid \left(1/\sqrt{d}, \mathbf{x}^{\top} \right)^{\top}\in  V\right\}\subset \mathbb{C}^{d^2-1}$ given by 
\begin{equation}
\label{eq:x}
\begin{array}{cccc}
\widehat{T} _0: & V _0 \times D_n & \longrightarrow & V _0\\
    	&(\mathbf{x}, {\bf z}) & \longmapsto & p({\bf z})\textbf{x}+q({\bf z}),
\end{array}
\end{equation}
with readout $\widehat{h} _0(\mathbf{x})= \widehat{h} \left(\left(1/\sqrt{d}, \mathbf{x}^{\top} \right)^{\top}\right)$. The system isomorphism  is in this case, given by the map 
\begin{equation}
\label{isomor v vo}
\begin{array}{cccc}
i _0: &V _0  &\longrightarrow &V \\
	&\mathbf{x} \in V _0 &\longmapsto &\left(1/\sqrt{d}, \mathbf{x}^{\top} \right)^{\top}. 
\end{array}
\end{equation}
Part (iv) of Proposition \ref{prop:isomorphism} guarantees that the dynamical properties of the system $(\widehat{T}, \widehat{h}, V) $ (and hence those of the QRC system $(T,h, \mathcal{S}({\mathcal H}))$)
are  equivalent to those of $(\widehat{T}_0, \widehat{h}_0, V_0) $.

The importance of this observation is that it links $(T,h, \mathcal{S}({\mathcal H}))$ to the non-homogeneous state-affine system (SAS) introduced in \cite{grigoryeva2018universal} and for which various universality properties have been additionally proved in \cite{gonon2019reservoir, RC13}. SAS are defined as state equations that have the form spelled out in \eqref{eq:x}. Strictly speaking, the SAS systems studied in the above-cited references impose polynomial or trigonometric dependences of $q$ and $p$ on the inputs, while in our situation, the prescription introduced in \eqref{Kraus with inputs} is capable of accommodating more general forms.

It has been shown in \cite{grigoryeva2018universal} that when such a system has the ESP, the corresponding filter $ U_{\widehat{T}_0}: \left(D_n\right)^{\mathbb{Z}} \rightarrow  \left(V _0\right)^{\mathbb{Z}}$ can be written as
\begin{equation} 
\label{eq:filter_x}
 U_{\widehat{T}_0}({\bf z})_t = \sum^{\infty}_{j=0}\left(\prod^{j-1}_{k=0}p({\bf z}_{t-k})\right)q({\bf z}_{t-j}),
\end{equation}
where 
$\prod_{k=0}^{j-1}p({\bf z}_{t-k}):=p({\bf z}_{t}) \cdot p({\bf z}_{t-1}) \cdots p({\bf z}_{t-j+1})$.
A first sufficient condition for ESP and FMP has been formulated in 
\cite{grigoryeva2018universal} by imposing that $\sigma_{\text{max}}(p({\bf z}))<1$ for all ${\bf z} \in D_n$, where $\sigma_{\text{max}}$ is the maximum singular value of matrix $p({\bf z})$. Given a matrix $A$, the maximum singular value is equal to the 2-Schatten induced norm:  $\vertiii{A}_2=\sigma_\text{max}(A)$, where the singular values of $A$ are the square-roots of the eigenvalues of $AA^{\dagger}$. An improved sufficient condition could be potentially found using other matrix norms as in \cite{buehner2006tighter}. 

The next proposition uses this hint and, moreover, spells out equivalent {\it necessary and sufficient} conditions for the ESP and FMP to hold in the three different equivalent representations for QRC systems that we have introduced in this section. More explicitly, the statement addresses the ESP and the FMP for the operator representation $T: \mathcal{B}(\mathcal{H}) \times D_n\rightarrow \mathcal{B}(\mathcal{H})$, the matrix representation $\widehat{T} :\mathbb{C}^{d^2} \times D _n \longrightarrow \mathbb{C}^{d^2} $ associated to the Gell-Mann basis $\mathcal{B} $, and the SAS representation  $\widehat{T}_0 :\mathbb{C}^{d^2-1} \times D _n \longrightarrow \mathbb{C}^{d^2-1} $ introduced in \eqref{eq:x}. 

\begin{proposition}
\label{prop:sufficient}
Let $T: \mathcal{B}(\mathcal{H}) \times D_n\rightarrow \mathcal{B}(\mathcal{H})$ be a continuous QRC system. The following three statements are equivalent:
\begin{enumerate}[(i)]
\item There exists an operator norm $\vertiii{\cdot } $ and $\epsilon>0 $ such that 
\begin{equation}
\label{contraction op1}
\vertiii{T(\cdot  , {\bf z})|_{\mathcal{B}_0(\mathcal{H})}}<1- \epsilon, \  \mbox{for all ${\bf z} \in D_n$}.
\end{equation}
\item There exists a matrix norm $\vertiii{\cdot }$ in the space of complex $d^2\times  d^2 $ matrices such that 
\begin{equation}
\label{contraction op2}
\vertiii{\widehat{T}(\cdot  , {\bf z})|_{G _{\mathcal{B}} ^{-1}\left(\operatorname{span} \left\{\mathcal{B} _0\right\}\right)}}<1- \epsilon, \  \mbox{for all ${\bf z} \in D_n$},
\end{equation}
with $\mathcal{B} _0 $ the trace-zero elements in the basis $\mathcal{B}  $ defined in \eqref{zero trace gens} and $G _{\mathcal{B}} $ the isomorphism defined in \eqref{system isom linear} with respect to the Gell-Mann basis.
\item There is a matrix norm $\vertiii{\cdot }$ in the space of complex $(d^2-1)\times  (d^2-1) $ matrices such that the SAS representation $\widehat{T}_0 :\mathbb{C}^{d^2-1} \times D _n \longrightarrow \mathbb{C}^{d^2-1} $ introduced in \eqref{eq:x} satisfies that
\begin{equation}
\label{contraction op3}
\vertiii{p({\bf z})}<1- \epsilon, \  \mbox{for all ${\bf z} \in D_n$},
\end{equation}
\end{enumerate}
If any of these three equivalent statements hold and $D_n$ is compact, then:
\begin{enumerate}
\item The isomorphic systems $T: \mathcal{S}(\mathcal{H}) \times D_n\rightarrow \mathcal{S}(\mathcal{H}) $, $\widehat{T }: V \times D_n\rightarrow V $, and $\widehat{T} _0: V_0 \times D_n\rightarrow V_0  $, have the ESP and the FMP, and hence continuous filters $U _T: (D _n)^{\mathbb{Z}} \longrightarrow \left(\mathcal{S}(\mathcal{H})\right)^{\mathbb{Z}}$, $U _{\widehat{T}}: (D _n)^{\mathbb{Z}} \longrightarrow V ^{\mathbb{Z}}$, and $U _{\widehat{T}_0}: (D _n)^{\mathbb{Z}} \longrightarrow V _0 ^{\mathbb{Z}}$ can be associated to them.
\item In such case, the filter $U_{\widehat{T}_0} $ is then given by \eqref{eq:filter_x}. $U _{\widehat{T}} $ is determined by $U _{\widehat{T}} = {\cal I}_0\circ U_{\widehat{T}_0} $, with ${\cal I}_{0}=\prod _{\Bbb Z} i _0 $ and $i _0 $ as in \eqref{isomor v vo}. Finally, $U _T={\cal G}_{\mathcal{B}} \circ U_{\widehat{T}} $, with ${\cal G}_{\mathcal{B}}=\prod _{\Bbb Z} G _{\mathcal{B}} $ the map introduced in Proposition \ref{prop:isomorphism}. 
\item  The contraction conditions  \eqref{contraction op1}-\eqref{contraction op3} are necessary for the ESP and the FMP to hold.
\end{enumerate}
\end{proposition}

\begin{proof}
The equivalences between the contraction conditions \eqref{contraction op1}-\eqref{contraction op3} can be shown by using norms in $\mathbb{C}^{d ^2 } $ and $\mathcal{B}({\mathcal H}) $ that make the map $G_{\mathcal{B}}$ in \eqref{system isom linear} into an isometry. More explicitly, let us start with an operator norm $\vertiii{\cdot } $ for which the maps $T(\cdot , {\bf z}): \mathcal{B}({\mathcal H})\longrightarrow \mathcal{B}({\mathcal H}) $ satisfy the condition \eqref{contraction op1}. Assume that this operator norm is associated with a given norm $\left\|\cdot \right\|_{\mathcal{B}({\mathcal H})} $ in $\mathcal{B}({\mathcal H}) $, that is, $\vertiii{T(\cdot , {\bf z})}=\sup_{A\neq 0} \left\{\left\|T(A, {\bf z})\right\|_{\mathcal{B}({\mathcal H})}/ \left\|A \right\|_{\mathcal{B}({\mathcal H})}\right\}$. Take now the norm $\left\|\cdot \right\|_{\mathbb{C}^{d^2}} $ in $\mathbb{C}^{d^2} $ with respect to which $G_{\mathcal{B}} $ is an isometry, that is, set $\left\|\mathbf{a} \right\|_{\mathbb{C}^{d^2}} = \left\|G_{\mathcal{B}}(\mathbf{a})\right\|_{\mathcal{B}({\mathcal H})}$, for all $\mathbf{a} \in \mathbb{C}^{d^2} $, and denote by $\vertiii{\widehat{T}(\cdot , {\bf z})}'=\sup_{\mathbf{a}\neq 0} \left\{\left\|\widehat{T}(\mathbf{a}, {\bf z})\right\|_{\mathbb{C}^{d^2}}/ \left\|\mathbf{a} \right\|_{\mathbb{C}^{d^2}}\right\} $. Using just these definitions, it is easy to see that $\vertiii{T(\cdot  , {\bf z})}=\vertiii{\widehat{T}(\cdot  , {\bf z})}' $ and, moreover, that
\begin{equation*}
\vertiii{T(\cdot  , {\bf z})|_{\mathcal{B}_0(\mathcal{H})}}=\vertiii{\widehat{T}(\cdot  , {\bf z})|_{G _{\mathcal{B}} ^{-1}\left(\operatorname{span} \left\{\mathcal{B} _0\right\}\right)}}',
\end{equation*}
which can be easily used to prove the equivalence between \eqref{contraction op1} and \eqref{contraction op2}. The equivalence between \eqref{contraction op2} and \eqref{contraction op3} follows from the fact that, as it can be seen in \eqref{matrix form t hat}, 
$$\widehat{T}(\cdot  , {\bf z})|_{G _{\mathcal{B}} ^{-1}\left(\operatorname{span} \left\{\mathcal{B} _0\right\}\right)}=p({\bf z}). $$

The claims $1$ and $2$ in the second part of the proposition are straightforward consequences of Propositions \ref{prop:contractivity implies ESP and FMP} and \ref{prop:isomorphism}. We now prove the last claim in point $3$ using the SAS representation \eqref{eq:x} for the continuous QRC system $T$ that, this time, it is assumed to have the ESP and the FMP. If we iterate $n$ times this state equation using an arbitrary vector $\mathbf{x} _0\in \mathbb{C}^{d ^2-1}$ as an initial condition, we obtain the vector $U^{n,\mathbf{x}_0} _{\widehat{T}_0}({\bf z})_t\in \mathbb{C}^{d ^2-1} $ defined by:
\begin{equation*}
 U^{n,\mathbf{x}_0} _{\widehat{T}_0}({\bf z})_t=\sum^{n-1}_{j=0}\left(\prod^{j-1}_{k=0}p({\bf z}_{t-k})\right)q({\bf z}_{t-j}) + \prod^{n-1}_{j=0}p({\bf z}_{t-j}) \mathbf{x}_0.
\end{equation*}
If we now assume that the system has the ESP, we necessarily have that 
\begin{equation}
\label{almost for ps}
U _{\widehat{T}_0}({\bf z})_t=\lim_{n \rightarrow \infty} U^{n,\mathbf{x}_0} _{\widehat{T}_0}({\bf z})_t=\lim_{n \rightarrow \infty} U^{n,\mathbf{x}_0'} _{\widehat{T}_0}({\bf z})_t,
\end{equation}
with $U _{\widehat{T}_0} $ the filter in \eqref{eq:filter_x} and $\mathbf{x} _0,\mathbf{x} _0' \in \mathbb{C}^{d ^2-1}$ arbitrary vectors. The last equality in \eqref{almost for ps} and the arbitrary character of $\mathbf{x} _0,\mathbf{x} _0' \in \mathbb{C}^{d ^2-1}$ imply that 
\begin{equation*}
\prod^{\infty}_{j=0}p({\bf z}_{t-j})= \lim_{n \rightarrow \infty}\prod^{n-1}_{j=0}p({\bf z}_{t-j})=0,
\end{equation*}
necessarily. We now notice that since $T$ is a continuous map, then so is the dependence of $p({\bf z})$ on the inputs ${\bf z} $. Moreover, since, in this case, $D _n $ is assumed to be a compact set, then so is the matrix set 
$\left\{p( {\bf z})\mid {\bf z} \in D _n\right\}$. Now, Corollary 6.4 in \cite{hartfiel2002nonhomogeneous} guarantees the existence of matrix norm $\vertiii{\cdot }$ for which \eqref{contraction op3} is satisfied, as required.
\end{proof}

\begin{remark}
\normalfont
Using the characterization of the mixing property in \eqref{condition to be mixing}, it is clear that a necessary condition for the QRC $T$ to satisfy the contractivity hypothesis in the previous proposition and, in passing, satisfy the ESP and the FMP, is that all the maps $T(\cdot , {\bf z}) $  are mixing for all ${\bf z} \in D _n $.  Indeed, since the spectral radius is a lower bound for any matrix norm (see \cite[Theorem 5.6.9]{horn2012matrix}) we have that $\lambda_{\text{max}}\left(\widehat{T}(\cdot  , {\bf z})|_{G _{\mathcal{B}} ^{-1}\left(\operatorname{span} \left\{\mathcal{B} _0\right\}\right)}\right)\leq \vertiii{\widehat{T}(\cdot  , {\bf z})|_{G _{\mathcal{B}} ^{-1}\left(\operatorname{span} \left\{\mathcal{B} _0\right\}\right)}}$.  Consequently, if $\vertiii{\widehat{T}(\cdot  , {\bf z})|_{G _{\mathcal{B}} ^{-1}\left(\operatorname{span} \left\{\mathcal{B} _0\right\}\right)}}< 1- \epsilon  $ then  $\lambda_{\text{max}}\left(\widehat{T}(\cdot  , {\bf z})|_{G _{\mathcal{B}} ^{-1}\left(\operatorname{span} \left\{\mathcal{B} _0\right\}\right)}\right)\leq 1- \epsilon  $, which is equivalent to 
$\lambda_{\text{max}}\left(T(\cdot , {\bf z})|_{\mathcal{B}_0(\mathcal{H})}\right)\leq 1- \epsilon  $
and then all maps $T(\cdot , {\bf z}) $ are mixing by  \eqref{condition to be mixing}. We emphasize that this mixing condition is necessary but not sufficient since, as the composition of two mixing maps is not necessarily mixing, the existence of the limit in~\eqref{eq:filter_rho} is not guaranteed even if each of the factor operators is mixing.
\end{remark}

\subsection{Some constrains on CPTP maps for QRC}
\label{Some constrains on CPTP maps for QRC}

The SAS representation allows us to easily characterize situations like the one in Proposition \ref{prop:sufficient} in which a QRC system has the ESP and the FMP and, moreover, it allows us to write explicitly down the corresponding filter \eqref{eq:filter_x}. As we shall now see in this section, more interesting facts can be derived from this representation having to do with design features that should be avoided, as they produce systems with only trivial solutions. The first one concerns {\it unital quantum channels}, that is, channels that satisfy 
$$T(I,{\bf z})=I \quad \mbox{for all ${\bf z} \in D_n$}.$$ 
That situation is studied in the next theorem, which will be generalized in Theorem \ref{th:constant} to the case of QRC systems that exhibit an input-independent fixed point.

\begin{theorem} 
\label{th:unital}
Let $T: \mathcal{B}(\mathcal{H}) \times D_n\rightarrow \mathcal{B}(\mathcal{H})$ be a QRC system for which there exists an operator norm $\vertiii{\cdot } $ and $\epsilon>0 $ such that 
$
\vertiii{T(\cdot  , {\bf z})|_{\mathcal{B}_0(\mathcal{H})}}<1- \epsilon$,  for all ${\bf z} \in D_n$. Then, the corresponding filter $U _T  $ is constant with $U_T({\bf z})_t=I/d$ for all ${\bf z}\in (D_n)^{\mathbb{Z}}$ (equivalently $U_{\widehat{T} _0}({\bf z})_t={\bf 0}$) if and only if $T$ is unital.
\end{theorem} 

\begin{proof}  
We prove this statement by using the SAS representation associated to the Gell-Mann basis $\mathcal{B} $. If $T$ is unital, that is, $T(I,{\bf z})=I$, then the expressions \eqref{eq:Wmatrix} and \eqref{matrix form t hat} imply that  $q({\bf z})_i=\tr(B_iT(I,{\bf z}))/\sqrt{d}=\tr(B_i)/\sqrt{d}=0$, for all $i \in \left\{2, \ldots, d ^2\right\} $. Therefore, the SAS state equation becomes, in this case, $\textbf{x}_t=p({\bf z}_t)\textbf{x}_{t-1}$, which is a homogeneous equation whose only solution under the contractivity hypotheses is the trivial one, that is, $\textbf{x}_t= {\bf 0}$ for all $t \in \Bbb Z$. Consequently, $U _{\widehat{T} _0}= {\bf 0} $. Proposition \ref{prop:sufficient}  implies then that $U _{T} ={\cal G}_{\mathcal{B}} \circ {\cal I}_0\circ U_{\widehat{T}_0}= I/d$.

Conversely, suppose that $U_T({\bf z})_t=I/d$ for all ${\bf z}\in (D_n)^{\mathbb{Z}}$. Since the filter $U _T $ is determined by the recursions
$U_T({\bf z})_t=T \left(U_T({\bf z})_{t-1}, {\bf z}_t\right)$,
then, this relation implies that $T(I,{\bf z})=I$, for all ${\bf z}\in D_n$.
\end{proof} 

\begin{remark}
\normalfont
There is a close relation between unital quantum channels and contractivity. Indeed, it has been shown that unital maps are contractive for all Schatten $p$-norms, that is, $\vertiii{T(\cdot,{\bf z})}_p\leq 1$ (see Theorem 2.4 in \cite{perez2006contractivity}) when $T(\cdot,{\bf z}) $ is unital. 
\end{remark}  

\begin{example}
\label{ex:dep} 
\normalfont
Consider the depolarizing channel $\mathcal{E}: \mathcal{S}({\mathcal H}) \longrightarrow \mathcal{S}({\mathcal H}) $ defined by 
\begin{equation} \label{eq:dep_S}
    \mathcal{E}(\rho)=(1-\lambda)\rho+\lambda\frac{I}{d},
\end{equation}
where $0\leq\lambda\leq 1$ denotes the probability of finding the system at the maximally mixed state $I/d$. If we arrange the previous equation as an input-dependent channel, we can write the following state equation:
\begin{equation}
\label{eq:dep_Ssta}
    \rho_t=\mathcal{E}(\rho_{t-1},{\bf z}_t)=(1-\lambda_t)\rho_{t-1}+\lambda_t\frac{I}{d},
\end{equation}
where $\lambda_t:=\lambda({\bf z}_t)$ is a function of the input at each time step. This map is strictly contractive whenever $r=\sup_{t\in \bb{Z}}(1-\lambda_t)<1$ in which case $r$ is the contraction rate. Let us compute the state after $n$ backwards iterations: 
\begin{equation}\label{eq:rho_n_dep}
\begin{split}
    \rho^n_t&=\left(\prod^{n-1}_{i=0}(1-\lambda_{t-i})\right)\rho_{t-n}\\
    &+\sum^{n-1}_{i=0}\left(\lambda_{t-i}\left(\prod^{i-1}_{j=0}(1-\lambda_{t-j})\right)\right)\frac{I}{d}.
\end{split}
\end{equation}
Taking the limit $n\rightarrow \infty$, it is easy to see that the first summand vanishes because $1-\lambda_{t-i}\leq r<1$ for all $t\in \bb{Z}$, and hence all possible dependence on the initial condition disappears. Regarding the second summand, since the map is strictly contractive, the limit $\lim_{n\rightarrow \infty} \rho^n_t$ exists and must yield a density matrix, which means that $\sum^{\infty}_{i=0}\left(\lambda_{t-i}\left(\prod^{i-1}_{j=0}(1-\lambda_{t-j})\right)\right)=1$ {because of the normalization}. This limit hence defines the filter $U_{\mathcal{E}}({\bf z})_t=I/d$, which is consistent with the conclusion of Theorem \ref{th:unital} since \eqref{eq:dep_Ssta} is a strictly contractive unital map.

We could have found the same result by directly using~\eqref{eq:x}, for which we need to define the extension of the depolarizing channel to the whole space of bounded operators $\mathcal{B}(\mathcal{H})$. We define then the CPTP map $\mathcal{E}': \mathcal{B}({\mathcal H}) \longrightarrow \mathcal{B}({\mathcal H}) $
\begin{equation} \label{eq:dep_B}
    \mathcal{E}'(A)=(1-\lambda)A+\lambda\tr(A)\frac{I}{d}.
\end{equation}
The SAS representation associated with the input-dependent version of this map is such that  $q({\bf z}_t)=0$ because the map is unital. It is also easy to see that the system equation becomes $\textbf{x}_t=p({\bf z}_t)\textbf{x}_{t-1}=(1-\lambda_t)\textbf{x}_{t-1}$. The filter $U_{\widehat{\mathcal{E}}' _0}$ of this QRC equation satisfies that $U_{\widehat{\mathcal{E}}' _0}({\bf z})_t={\bf 0}$, for all $t \in \Bbb Z$. 
\end{example}

\begin{example}
\label{ex:ing}
\normalfont
The next case is an example of ``poorly engineered"  QRC system which can be detected using Theorem \ref{th:unital}. Indeed, we shall introduce a model of dissipation using tuneable local losses, but Theorem \ref{th:unital} will discard its long-term applicability because it is unital.

Let us define the Markovian master equation that governs the dynamics between input injections:
\begin{equation}
\label{eq:me_1qubit0}
    \dot{\rho}=-i[H({\bf z}_t),\rho]+\gamma L\rho L^{\dagger}-\frac{\gamma}{2}\{L^{\dagger}L,\rho\}, 
\end{equation}
where $L$ is the jump operator and $\{L^{\dagger}L,\rho\}$ denotes the anticommutator. We define the input-dependent Hamiltonian as $H({\bf z}_t)=h({\bf z}_t)\sigma^x/2$, where $h({\bf z}_t)$ will be an arbitrary function of the input, and the jump operator as $L=\sigma^z$. This is a single qubit under the influence of an external magnetic field in the $x$ direction of the real space (whose intensity varies between inputs) with a local dephasing. Notice that the Hamiltonian is considered as time-independent when integrating the dynamics since it is constant between input injections. 

Going from the density matrix language to the real variable linear description requires to find the CPTP map representation of~\eqref{eq:me_1qubit0}. Since Markovian master equations are CPTP linear transformations on their own, we just need to find the Kraus decomposition that represents the dynamics of this master equation as a map. We will follow the procedure as explained in \cite{andersson2007finding} (see Section 2 in the reference for the details of the algorithm). In particular, the Kraus decomposition of a single qubit can be written in the following form:
\begin{equation}\label{eq:kraus_1qubit}
    T(\rho)=\sum_{i,j}S^{(U)}_{ij}B_i\rho B^{\dagger}_j,
\end{equation}  
where $B_i$ are the basis elements of a single qubit in the operator space ($\{I,\sigma^x,\sigma^y,\sigma^z\}$) and $S^{(U)}_{ij}:=U^{\dagger}SU$ is a unitary transformation of the Choi matrix $S$. Applying~\eqref{eq:kraus_1qubit} to the definition of matrix $T$ in~\eqref{eq:Wmatrix}, we can find the map of the single qubit observables:
\begin{equation}
  \left( \begin{matrix} 1   \\ 
  \braket{\sigma^x}_t \\
  \braket{\sigma^y}_t \\
  \braket{\sigma^z}_t
  \end{matrix} \right) =\left( \begin{matrix} 1 & 0 & 0 & 0  \\ 
 
   0 & \widehat{T}_{22} & 0 & 0 \\
   0 & 0 &  \widehat{T}_{33} & \widehat{T}_{34} \\
   0 &  0 &  \widehat{T}_{43} & \widehat{T}_{44}
  \end{matrix} \right) \left( \begin{matrix} 1   \\ 
  \braket{\sigma^x}_{t-1} \\
  \braket{\sigma^y}_{t-1} \\
  \braket{\sigma^z}_{t-1}
  \end{matrix} \right),
\end{equation}
where $\braket{\sigma^a}:=\tr(\sigma^a\rho)$ is the expected value of the spin projection in the $a$ direction of the real space. The expressions for each matrix element are shown below:
\begin{widetext}
\begin{equation}
    \begin{matrix}
        \widehat{T}_{22}=e^{-2\gamma \Delta \tau }, \\
        \widehat{T}_{33} = e^{-\gamma \Delta \tau }\left(\cosh\left(\Delta \tau \sqrt{\gamma^2-h_t^2}\right)-\frac{\gamma}{\sqrt{\gamma^2-h_t^2}}\sinh\left(\Delta \tau \sqrt{\gamma^2-h_t^2}\right)\right),  \\
       \widehat{T}_{44} =  e^{-\gamma \Delta \tau }\left(\cosh\left(\Delta \tau \sqrt{\gamma^2-h_t^2}\right)+\frac{\gamma}{\sqrt{\gamma^2-h_t^2}}\sinh\left(\Delta \tau \sqrt{\gamma^2-h_t^2}\right)\right), \\
        \widehat{T}_{34} = \frac{h_te^{-\gamma \Delta \tau }}{\sqrt{\gamma^2-h_t^2}}\sinh\left(\Delta \tau \sqrt{\gamma^2-h_t^2}\right),\\
        \widehat{T}_{43} = - \widehat{T}_{34}, 
    \end{matrix}
\end{equation}
\end{widetext}
where have shortened notation by setting $h_t=h({\bf z}_t)$.
The eigenvalues of matrix $\widehat{T}$ can be computed analytically: $\lambda_1=1$, $\lambda_2=e^{-2\gamma \Delta \tau }$, $\lambda_3=e^{-(\gamma+\sqrt{\gamma^2-h^2_t})\Delta \tau }$ and $\lambda_4=e^{-(\gamma-\sqrt{\gamma^2-h^2_t})\Delta \tau }$. The moduli of the eigenvalues are $|\lambda_1|=1$, $|\lambda_2|=e^{-2\gamma \Delta \tau }<1$ and $|\lambda_3|=e^{-(\gamma+\sqrt{\gamma^2-h^2_t})\Delta \tau }<1$. Eigenvalue $|\lambda_4|=e^{-(\gamma-\sqrt{\gamma^2-h^2_t})\Delta \tau }$ is smaller than one if and only if $h_t\neq 0$. Under that condition, the map is a mixing channel with a single fixed point. It can be checked that the map is unital so the fixed point is the maximally mixed state:
\begin{equation}
    \rho^*=\left( \begin{matrix} 1/2 & 0  \\ 0 & 1/2  \end{matrix} \right).
\end{equation}

Let us prove that there exists some norm where $\vertiii{p(\textbf{z})}<1$ for all inputs. A straightforward induced norm to evaluate is the 2-Schatten induced norm. The singular values of the restriction 
\begin{equation}
 p(\textbf{z})=\left( \begin{matrix} 
     \widehat{T}_{22} & 0 & 0 \\
    0 &  \widehat{T}_{33} & \widehat{T}_{34} \\
     0 &  \widehat{T}_{43} & \widehat{T}_{44}
  \end{matrix} \right)
\end{equation}
are $\sigma_1=e^{-2\gamma \Delta \tau }<1$, $\sigma_2=e^{-\gamma \Delta \tau }\sqrt{f_+}$ and $\sigma_3=e^{-\gamma \Delta \tau }\sqrt{f_-}$, where
\begin{widetext}
\begin{equation}
\begin{split}
    f_\pm &=\frac{1}{\gamma^2-h^2_t}\left(-h^2_t+\gamma^2\cosh\left(2\Delta \tau \sqrt{\gamma^2-h_t^2}\right)\right. \\
     &\left.\pm\gamma\sqrt{\sinh^2\left(\Delta \tau \sqrt{\gamma^2-h_t^2}\right)}\sqrt{-4h^2_t+2\gamma^2+2\gamma^2\cosh\left(2\Delta \tau \sqrt{\gamma^2-h_t^2}\right)}\right).      
\end{split}
\end{equation} 
\end{widetext}

\begin{figure}[h]
\captionsetup[subfigure]{}
\begin{center}
\includegraphics[scale=1.05]{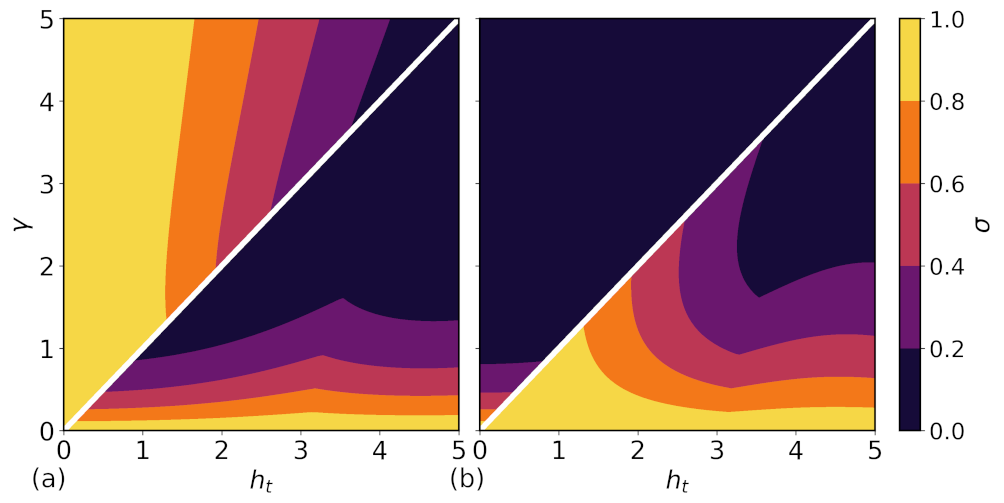}
\caption{Density plot of the singular values (a) $\sigma_2$ and (b) $\sigma_3$. The diagonal elements $h_t=\gamma$ are not determined because of the denominator $1/(\gamma^2-h^2_t)$.}\label{Fig:1}
\end{center}
\end{figure}

Figure \ref{Fig:1} numerically shows that for $\gamma\neq h_t$ and away from the axes $h_t=0$ and $\gamma=0$, we find $\sigma_2, \sigma_3<1$.
Therefore, the system has the ESP and the FMP. As we showed in Theorem \ref{th:unital}, the corresponding filter~\eqref{eq:filter_rho} is necessarily trivial and given by
\begin{equation}
    U _T({\bf z})_t = \rho^*=\left( \begin{matrix} 1/2 & 0  \\ 0 & 1/2  \end{matrix} \right).
\end{equation}
Since the Bloch vector for this constant matrix is  $(0,0,0)^{\top}$, this shows that in the SAS representation $U _{\widehat{T}_0}({\bf z})_t=(0,0,0)^{\top} $. 
\end{example}

Unital quantum channels are very common in the quantum information literature because of their practical advantages and well-known mathematical properties, see for example \cite{mendl2009unital,watrous2018theory}, and references therein. However, Theorem \ref{th:unital} discards the possibility of relying on unital contractive channels for QRC for long input sequences (see also Section \ref{sec:discussion} for more details). The physical explanation for this behavior stems from the fact that the fixed point of these maps does not depend on the input. Then, after each application of the channel, decoherence always leads to the same stationary state (the maximally mixed state), which does not keep track of these inputs. This hinders any possibility of storing the input information into the degrees of freedom of the quantum system since all coherences fade out and the diagonal elements of the density matrix become equal. 

We believe that the observation that we just made is very relevant, namely that quantum channels with input-independent fixed points become memoryless in the long-term. This fact is proved in the next theorem that generalizes Theorem \ref{th:unital}.

\begin{theorem} 
\label{th:constant}
Let $T: \mathcal{B}(\mathcal{H}) \times D_n\rightarrow \mathcal{B}(\mathcal{H})$ be a QRC system for which there exists an operator norm $\vertiii{\cdot } $ and $\epsilon>0 $ such that 
$
\vertiii{T(\cdot  , {\bf z})|_{\mathcal{B}_0(\mathcal{H})}}<1- \epsilon$,  for all ${\bf z} \in D_n$. Then, $T$ has an input-independent fixed point $\rho^\ast  \in \mathcal{S} ({\mathcal H})$, that is, $T(\rho^*,{\bf z})=\rho^*$, for all ${\bf z}\in D_n$, if and only if the corresponding filter $U _T  $ is constant, that is, $U_T({\bf z})_t=\rho^\ast  \in \mathcal{S} ({\mathcal H})$, for all ${\bf z}\in (D_n)^{\mathbb{Z}}$ (equivalently, $U_{\widehat{T} _0}({\bf z})_t= i _0^{-1}\left(G_{\mathcal{B}}^{-1} (\rho^\ast )\right)$) .
\end{theorem} 

\begin{proof}  
We first note that the contractivity hypothesis implies by Proposition \ref{prop:sufficient} that the system associated to $T$ has the ESP and hence has a unique solution for each input. The hypothesis $T(\rho^*,{\bf z})=\rho^*$, for all ${\bf z}\in D_n$, obviously implies that the constant sequence equal to $\rho^* $ is a solution for any input ${\bf z} \in (D _n)^{\mathbb{Z}} $ and hence  $U_T({\bf z})_t=\rho^\ast  \in \mathcal{S} ({\mathcal H})$, for all ${\bf z}\in (D_n)^{\mathbb{Z}}$. Conversely, since the filter $U _T $ is determined by the recursions
\begin{equation}
\label{intermediate U}
U_T({\bf z})_t=T \left(U_T({\bf z})_{t-1}, {\bf z}_t\right),
\end{equation}
then, if we always have that $U_T({\bf z})_t=\rho^\ast $, the relation \eqref{intermediate U} implies that $T(\rho^*,{\bf z})=\rho^*$, for all ${\bf z}\in D_n$.
\end{proof}

\begin{remark}
\normalfont
The differences between the hypotheses in Theorems \ref{th:unital} and \ref{th:constant} on fixed points are apparent when the QRC system is expressed using the SAS representation in terms of the functions $q $ and $p$. Indeed, as we saw in the proof of Theorems \ref{th:unital}, $q({\bf z})=0$ for and ${\bf z} \in D_n $, in that case, and the input dependence takes place only through $p$. This is the case in Example \ref{ex:dep}. Theorem  \ref{th:constant} allows for an input dependence through $q$ too.  
\end{remark}

\begin{example}
\label{ex:bad}
\normalfont    
We define a Markovian master equation for the dynamics between input injections:
\begin{equation}
    \dot{\rho}=-i[H({\bf z}_t),\rho]+\gamma L\rho L^{\dagger}-\frac{\gamma}{2}\{L^{\dagger}L,\rho\}, 
\end{equation}
where $H({\bf z}_t)=h({\bf z}_t)\sigma^z/2$ and $L=\sigma^-$. This is a single qubit under the influence of an external magnetic field in the $z$ direction with local dissipation.
 The  matrix expression $\widehat{T}$ for the associated system is:
 \begin{widetext}
\begin{equation}
  \left( \begin{matrix} 1   \\ 
  \braket{\sigma^x}_t \\
  \braket{\sigma^y}_t \\
  \braket{\sigma^z}_t
  \end{matrix} \right)  
  =\left( \begin{matrix} 1 & 0 & 0 & 0  \\ 
   0 & e^{-\frac{\gamma \Delta \tau }{2}}\cos(h_t\Delta \tau ) & e^{-\frac{\gamma \Delta \tau }{2}}\sin(h_t\Delta \tau ) & 0 \\
   0 & -e^{-\frac{\gamma \Delta \tau }{2}}\sin(h_t\Delta \tau ) & e^{-\frac{\gamma \Delta \tau }{2}}\cos(h_t\Delta \tau ) & 0 \\
   e^{-\gamma \Delta \tau }-1 &  0 &  0 & e^{-\gamma \Delta \tau }
  \end{matrix} \right) \left( \begin{matrix} 1   \\ 
  \braket{\sigma^x}_{t-1} \\
  \braket{\sigma^y}_{t-1} \\
  \braket{\sigma^z}_{t-1}
  \end{matrix} \right).    
\end{equation}
\end{widetext}
The eigenvalues of  $\widehat{T}$ can be computed analytically: $\lambda_1=1$, $\lambda_2=e^{-\gamma \Delta \tau }$, $\lambda_3=e^{-\frac{\gamma \Delta \tau }{2}-ih_t\Delta \tau }$ and $\lambda_4=e^{-\frac{\gamma \Delta \tau }{2}+ih_t\Delta \tau }$. The moduli of the eigenvalues are $|\lambda_1|=1$, $|\lambda_2|=e^{-\gamma \Delta \tau }<1$ and $|\lambda_3|=|\lambda_4|=e^{-\frac{\gamma \Delta \tau }{2}}<1$, so $T$ is a mixing channel.  In this case, the single fixed point is a pure input-independent state with density matrix 
\begin{equation}
    \rho^*=\left( \begin{matrix} 0 & 0  \\ 0 & 1  \end{matrix} \right).
\end{equation}

Let us see if it is true that there exists some $\vertiii{p(\textbf{z})}<1$ for all inputs. The singular values of the restriction 
\begin{equation}\label{eq:p}
 p(\textbf{z})=\left( \begin{matrix} 
     e^{-\frac{\gamma \Delta \tau }{2}}\cos(h_t\Delta \tau ) & e^{-\frac{\gamma \Delta \tau }{2}}\sin(h_t\Delta \tau ) & 0 \\
     -e^{-\frac{\gamma \Delta \tau }{2}}\sin(h_t\Delta \tau ) & e^{-\frac{\gamma \Delta \tau }{2}}\cos(h_t\Delta \tau ) & 0 \\  0 &  0 & e^{-\gamma \Delta \tau }
  \end{matrix} \right)
\end{equation}
are $\sigma_1=e^{-\gamma \Delta \tau }<1$ and $\sigma_2=\sigma_3=e^{-\frac{\gamma \Delta \tau }{2}}<1$. Therefore, the system has the ESP and FMP. As we showed in Theorem \ref{th:constant}, the associated filter is necessarily constant, and \eqref{eq:filter_rho} necessarily yields the filter 
\begin{equation}
U _T({\bf z})_t = \rho^*=\left( \begin{matrix} 0 & 0  \\ 0 & 1  \end{matrix} \right).
\end{equation}
Since the Bloch vector for this constant matrix is  $(0,0,-1)^{\top}$, this shows that in the SAS representation $U _{\widehat{T}_0}({\bf z})_t=(0,0,-1)^{\top} $. 

We can double-check the solution by explicitly computing the filter~\eqref{eq:filter_x}. Take $p({\bf z}_t)$ as in \eqref{eq:p} and $q({\bf z}_t)^{\top}=(0,0,e^{-\gamma \Delta \tau }-1)$. Since $\sigma_{\text{max}}(p({\bf z}_t))<1$, the filter in~\eqref{eq:filter_x} exists. Using now that $q({\bf z}_t)$ is input-independent, we can write 
\begin{equation}
\begin{split}
U_{\widehat{T} _0}({\bf z})_t &=\sum^{\infty}_{j=0}\left(\prod^{j-1}_{k=0}p({\bf z}_{t-k})\right)q({\bf z}_{t-j})\\
    &=\left(\sum^{\infty}_{j=0}\prod^{j-1}_{k=0}p({\bf z}_{t-k})\right)q=(e^{-\gamma \Delta \tau }-1)M_3,       
\end{split}
\end{equation}
where $M_3$ is the third column of the matrix $M=\sum^{\infty}_{j=0}\prod^{j-1}_{k=0}p({\bf z}_{t-k})$. Notice that the third column of the product $\prod^{j-1}_{k=0}p({\bf z}_{t-k})$ is $(0,0,e^{-k\gamma \Delta \tau })^{\top}$. Then, the column $M_3$ equals
\begin{equation}
    M_3 = \left(\begin{matrix} 0   \\ 0 \\ \sum^{\infty}_{j=0} e^{-k\gamma\Delta \tau }  \end{matrix} \right)=\left(\begin{matrix} 0   \\ 0 \\ \frac{1}{1-e^{-\gamma\Delta \tau }}  \end{matrix} \right),
\end{equation}
and it hence follows that $U _{\widehat{T}_0}({\bf z})_t=(0,0,-1)^{\top} $.

To conclude, we show in Figure \ref{Fig:bad} a numerical example of input driving. We choose $h(z_t)=z_t\frac{h}{2}\sigma^z$ as the external magnetic field function, where $h$ is a constant and $z_t$ is a unidimensional random input. The input will be drawn from a random uniform distribution in the interval $[0,1]$. As can be seen, the observables exhibit a transient time after which they converge to the input-independent stationary state. 

\begin{figure}[h]
\captionsetup[subfigure]{}
\begin{center}
\includegraphics[scale=0.9]{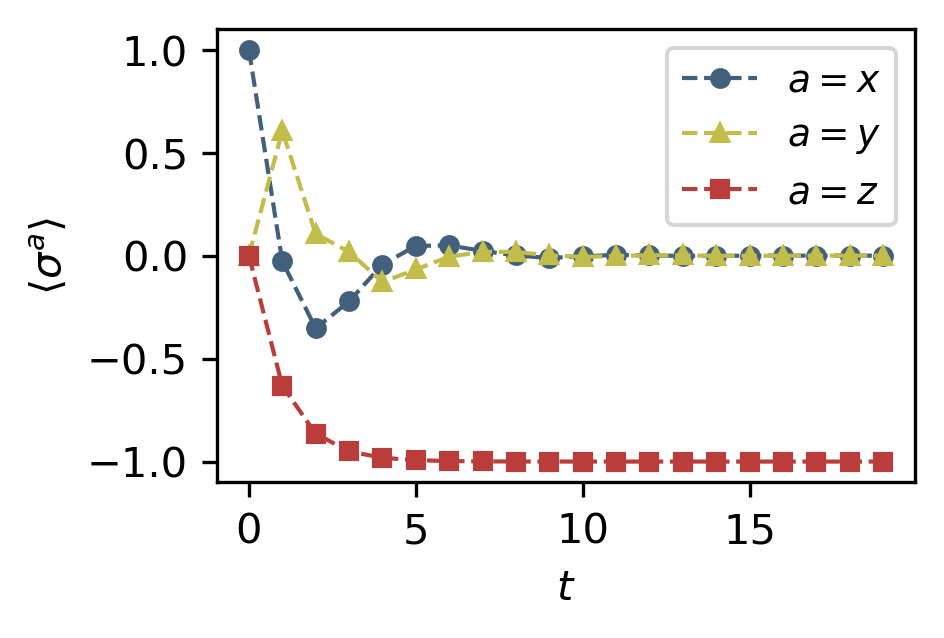}
\caption{Dynamics of the spin projections $\braket{\sigma^a}_t$ for $a=x,y,z$ when driven with a random input sequence with Hamiltonian $H=z_t\frac{h}{2}\sigma^z$. The initial condition is the maximal coherent state $\rho=1/2\sum^1_{i,j=0}\ket{i}\bra{j}$ and the system parameters are $\Delta \tau=1$, $\gamma=1$ and $h=1$.}\label{Fig:bad}
\end{center}
\end{figure}

\end{example}

\section{Discussion}\label{sec:discussion}
Theorems \ref{th:unital} and \ref{th:constant} bring up a connection between long-term computation and noisy intermediate-scale quantum (NISQ) devices: they have a finite time of operation due to decoherence. Consider a model of the type 
\begin{equation}\label{eq:deco}
    \rho_t=T(\rho _t,{\bf z} _t)=\mathcal{E}_{\text{deco}}(\mathcal{U}({\bf z}_t)\rho_{t-1}\mathcal{U}^{\dagger}({\bf z}_t)),
\end{equation}
where $\mathcal{E}_{\text{deco}}$ represents the decoherence produced by the contact of the system with an external environment. 
This model is present in QRC experimental works like \cite{kubota2022quantum,suzuki2022natural}, where the unitary dynamics $\mathcal{U}({\bf z}_t)$ is given by a quantum circuit. Theorem \ref{th:unital} explains what happens in the extreme case in which the quantum noise of the device is unital.
If the decoherence channel $\mathcal{E}_{\text{deco}}$ is a unital strictly contractive map then
$$T\left(\frac{I}{d},{\bf z} _t\right)=\mathcal{E}_{\text{deco}}\left(\mathcal{U}({\bf z}_t)\frac{I}{d}\mathcal{U}^{\dagger}({\bf z}_t)\right)=\frac{\mathcal{E}_{\text{deco}}(I)}{d}=\frac{I}{d},$$ 
that is, $T(\cdot,{\bf z})$ becomes a unital strictly contractive map for all ${\bf z}\in D_n$. Theorem \ref{th:unital} shows, in this case, that the filter becomes trivial after the injection of long input sequences. Instances of unital decoherence can be found in depolarizing channels, like in Example \ref{ex:dep}, or in dephasing channels. A dephasing channel damps the coherences of the density matrix but does not affect the diagonal elements. 

More generally, Theorem \ref{th:constant} can be interpreted as a ``common sense" warning: it explicitly states that the input codification must have a measurable influence on the attractor of the natural dynamics of the CPTP map. Otherwise, there is no possibility of storing input information in the long run. We emphasize that this is independent of the type of dissipation that the quantum channel produces. We can connect the implications of this theorem with the NISQ discussion and  \eqref{eq:deco}. If the input is codified in some particular coherences of the system, one must be careful that decoherence does not destroy those matrix elements, because then input-dependence would vanish, and the resulting filter would become trivial. 

This discussion does not imply that models like \eqref{eq:deco} are useless. Indeed, the opposite has been proven in previous works for short-term memory tasks \cite{kubota2022quantum,suzuki2022natural}. Theorems \ref{th:unital} and \ref{th:constant} only rigorously establish something that was already known about NISQ devices, that is, that there is a coherence time in which they can be exploited. Equivalently, models affected by these theorems have the ESP, but the resulting input/output dynamics becomes trivial for long input sequences. 

The coherence time limitation affects to all quantum platforms to a greater or lesser extent. Then, either QRC proposals are subject to operate on shorter time scales than the natural noise time scale (as done in \cite{kubota2022quantum,suzuki2022natural}), or the QRC system is carefully design to integrate it. For example, as we will see in Section \ref{ex:good}, dephasing can be integrated as part of a QRC system that is not affected by Theorems \ref{th:unital} and \ref{th:constant}.

We could further extend this analysis to QRC models with measurements. As an example we take the model proposed in \cite{mujal2022time}. In this reference, the quantum measurement is applied at each time step after the input dependent CPTP map $T$. The measurement scheme is introduced by modeling an indirect measurement with a continuous-variable ancilla \cite{naghiloo2019introduction,SiddiqiNature,SiddiqiQtrajectories,Lecocq2021,Hatridge2013}, producing a quantum reservoir with stochastic dynamics. For simplicity, we will restrict ourselves to the case of a single qubit, and we will average the quantum states over the limit of infinite measurements, yielding an unconditional state which is led by a deterministic CPTP quantum channel \cite{wiseman2009quantum}. Besides, we choose, without loss of generality, to take measurements in the $z$ direction. Under these conditions, the CPTP map is 
\begin{equation}
\label{eqrhoidealwithg}
\rho_t=M\odot T(\rho_{t-1},\textbf{z}_t),
\end{equation}
where $\odot$ represents the Hadamard or element-wise matrix product and $M$ is defined as
\begin{equation}\label{eq:M}
   M=\begin{pmatrix}
1&e^{-\frac{g^2}{2}}\\
e^{-\frac{g^2}{2}}&1
\end{pmatrix}.
\end{equation}
The measurement strength $g$ allows us to quantify the decoherence introduced by sharp measurements ($g\gg1$), while for $g\ll1$ the state is weakly perturbed. It is straightforward to see that this model is introducing dephasing, such that we can rewrite ~\eqref{eqrhoidealwithg} as 
\begin{equation}\label{eq:deph_T}
 \rho_t=\mathcal{E}_{\text{deph}}\left( T(\rho_{t-1},\textbf{z}_t)\right),  
\end{equation}
where the dephasing channel $\mathcal{E}_{\text{deph}}$ is defined as
\begin{equation}\label{eq:deph}
    \mathcal{E}_{\text{deph}}(\rho)=\frac{1+e^{-\frac{g^2}{2}}}{2}\rho+\frac{1-e^{-\frac{g^2}{2}}}{2}\sigma^z\rho\sigma^z.
\end{equation}
As we explained above, unitary dynamics for the map $T$ would lead to a memoryless reservoir in the long-term. However, one could engineer a mixing map $T$ such that there is a competition between the attractors of maps $T$ and $\mathcal{E}_{\text{deph}}$. The final fixed point of~\eqref{eq:deph_T} would be somewhere between the original fixed point of $T$ and a diagonal state (which is the shape of the fixed points of $\mathcal{E}_{\text{deph}}$).  

We conclude by presenting an example of a qubit with tunable local dissipation that fulfills all the requirements to be a ``properly engineered" QRC system. Then, we extend it with the measurement model of \cite{mujal2022time} to show that it still constitutes a proper QRC system.

\subsection{A ``properly engineered" QRC system} 
\label{ex:good}
We start by introducing the model without measurements. The Markovian master equation that governs the dynamics between input injections is:
\begin{equation}
    \dot{\rho}=-i[H({\bf z}_t),\rho]+\gamma L\rho L^{\dagger}-\frac{\gamma}{2}\{L^{\dagger}L,\rho\}, 
\end{equation}
where $H({\bf z}_t)=h({\bf z}_t)\sigma^x/2$, and $L=\sigma^-$. 
 The corresponding matrix  expression $\widehat{T}$ is
\begin{equation}\label{eq:T_good}
  \left( \begin{matrix} 1   \\ 
  \braket{\sigma^x}_t \\
  \braket{\sigma^y}_t \\
  \braket{\sigma^z}_t
  \end{matrix} \right) =\left( \begin{matrix} 1 & 0 & 0 & 0  \\ 
 
   0 & \widehat{T}_{22} & 0 & 0 \\
   \widehat{T}_{31} & 0 &  \widehat{T}_{33} & \widehat{T}_{34} \\
   \widehat{T}_{41} &  0 &  \widehat{T}_{43} & \widehat{T}_{44}
  \end{matrix} \right) \left( \begin{matrix} 1   \\ 
  \braket{\sigma^x}_{t-1} \\
  \braket{\sigma^y}_{t-1} \\
  \braket{\sigma^z}_{t-1}
  \end{matrix} \right),
\end{equation}
where the expressions for the matrix elements are shown below:
\begin{widetext}
\begin{equation}
    \begin{matrix}
        \widehat{T}_{22}=e^{-\frac{\gamma \Delta \tau }{2}}, \\
        \widehat{T}_{33} = e^{-\frac{3\gamma \Delta \tau }{4}}\left(\cosh\left(\frac{\Delta \tau }{4}\sqrt{\gamma^2-16h_t^2}\right)+\frac{\gamma}{\sqrt{\gamma^2-16h_t^2}}\sinh\left(\frac{\Delta \tau }{4}\sqrt{\gamma^2-16h_t^2}\right)\right),  \\
       \widehat{T}_{44} =  e^{-\frac{3\gamma \Delta \tau }{4}}\left(\cosh\left(\frac{\Delta \tau }{4}\sqrt{\gamma^2-16h_t^2}\right)-\frac{\gamma}{\sqrt{\gamma^2-16h_t^2}}\sinh\left(\frac{\Delta \tau }{4}\sqrt{\gamma^2-16h_t^2}\right)\right), \\
        \widehat{T}_{34} = \frac{4h_te^{-\frac{3\gamma \Delta \tau }{4}}}{\sqrt{\gamma^2-16h_t^2}}\sinh\left(\frac{\Delta \tau }{4}\sqrt{\gamma^2-16h_t^2}\right),\\
        \widehat{T}_{43} = - \widehat{T}_{34}, \\
        \widehat{T}_{31} = \frac{2\gamma h_t}{\gamma^2+2h^2_t}\left\{-1+e^{-\frac{3\gamma \Delta \tau }{4}}\left(\cosh\left(\frac{\Delta \tau }{4}\sqrt{\gamma^2-16h_t^2}\right)+\frac{3\gamma}{\sqrt{\gamma^2-16h_t^2}}\sinh\left(\frac{\Delta \tau }{4}\sqrt{\gamma^2-16h_t^2}\right)\right)\right\}, \\
        \widehat{T}_{41}= \frac{\gamma}{\gamma^2+2h^2_t}\left\{-\gamma +e^{-\frac{3\gamma \Delta \tau }{4}}\left(\gamma\cosh\left(\frac{\Delta \tau }{4}\sqrt{\gamma^2-16h_t^2}\right)-\frac{\gamma^2+8h^2_t}{\sqrt{\gamma^2-16h_t^2}} \sinh\left(\frac{\Delta \tau }{4}\sqrt{\gamma^2-16h_t^2}\right)\right)\right\}.
    \end{matrix}
\end{equation}
\end{widetext}

\begin{figure}[h]
\captionsetup[subfigure]{}
\begin{center}
\includegraphics[scale=1.05]{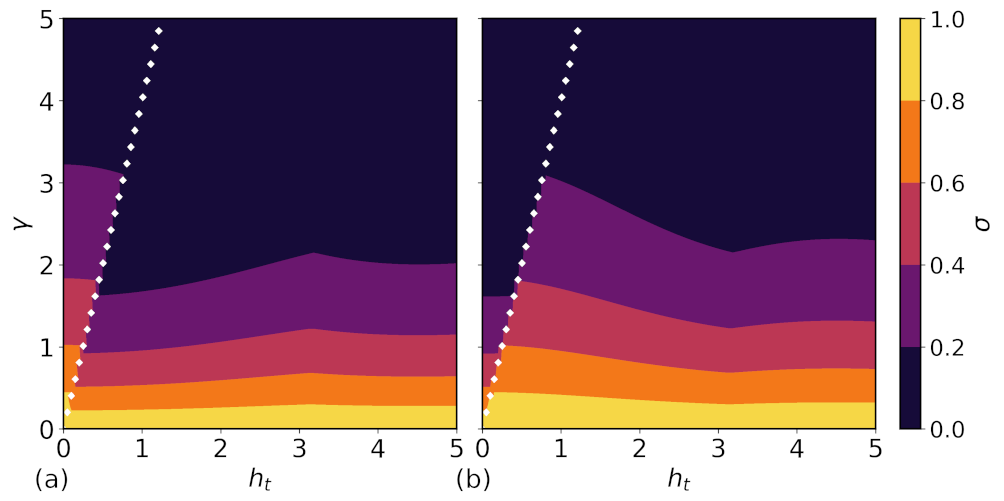}
\caption{Density plot of the singular values (a) $\sigma_2$ and (b) $\sigma_3$. The elements $\gamma=4h_t$ are not determined because of the denominator $1/(\gamma^2-16h^2_t)$.}\label{Fig:2}
\end{center}
\end{figure}
The eigenvalues of matrix $\hat{T}$ are: $\lambda_1=1$, $\lambda_2=e^{-\frac{\gamma \Delta \tau }{2}}$, $\lambda_3=e^{-(3\gamma^2+\sqrt{\gamma^2-16h^2_t})\frac{\Delta \tau }{4}}$ and $\lambda_4=e^{-(3\gamma^2-\sqrt{\gamma^2-16h^2_t})\frac{\Delta \tau }{4}}$. The three eigenvalues $\lambda_2$, $\lambda_3$ and $\lambda_4$ have always modulus smaller than one when $\gamma\neq 0$. Then, the master equation fulfills the conditions for having a single full-rank fixed point \cite{nigro2019uniqueness}, whose density matrix is
\begin{equation}\label{eq:fixed_good}
    \rho^*=\frac{1}{\gamma^2+2h^2_t} \left( \begin{matrix} h^2_t & i\gamma h_t  \\ -i\gamma h_t & \gamma^2 +h^2_t  \end{matrix} \right).
\end{equation}
Finally, we compute the singular values of the restriction to the traceless hyperplane. These values are $\sigma_1=e^{-\frac{\gamma \Delta \tau }{2}}<1$, $\sigma_2=e^{-\frac{3\gamma \Delta \tau }{4}}\sqrt{f_+}$ and $\sigma_3=e^{-\frac{3\gamma \Delta \tau }{4}}\sqrt{f_-}$, where
\begin{widetext}
\begin{equation}
\begin{split}
    f_\pm &=\frac{1}{\gamma^2-16h^2_t}\left(-16h^2_t+\gamma^2\cosh\left(\frac{\Delta \tau }{2}\sqrt{\gamma^2-16h_t^2}\right)\right. \\    
     &\left.\pm\gamma\sqrt{\sinh^2\left(\frac{\Delta \tau }{4}\sqrt{\gamma^2-16h_t^2}\right)}\sqrt{-64h^2_t+2\gamma^2+2\gamma^2\cosh\left(\frac{\Delta \tau }{2}\sqrt{\gamma^2-16h_t^2}\right)}\right).       
\end{split}
\end{equation}  
\end{widetext}
Figure \ref{Fig:2} shows that for $\gamma\neq 4h_t$ and away from the axis $\gamma=0$, we find $\sigma_2, \sigma_3<1$, demonstrating the ESP and the FMP. Since the fixed point $\rho^*$ is input dependent, this system exhibits the necessary ingredients to be a competent QRC system.

As in Example \ref{ex:bad}, we show in Figure \ref{Fig:good} a numerical example of input driving. We choose again $h(z_t)=z_t\frac{h}{2}\sigma^x$ as the external magnetic field function, modifying its direction. Now, the observables $\braket{\sigma^z}$ and  $\braket{\sigma^y}$ exhibit an explicit response to the driving, while $\braket{\sigma^x}$ converges to its input-independent stationary value (which can be predicted from \eqref{eq:T_good}). 
\begin{figure}[h]
\captionsetup[subfigure]{}
\begin{center}
\includegraphics[scale=0.9]{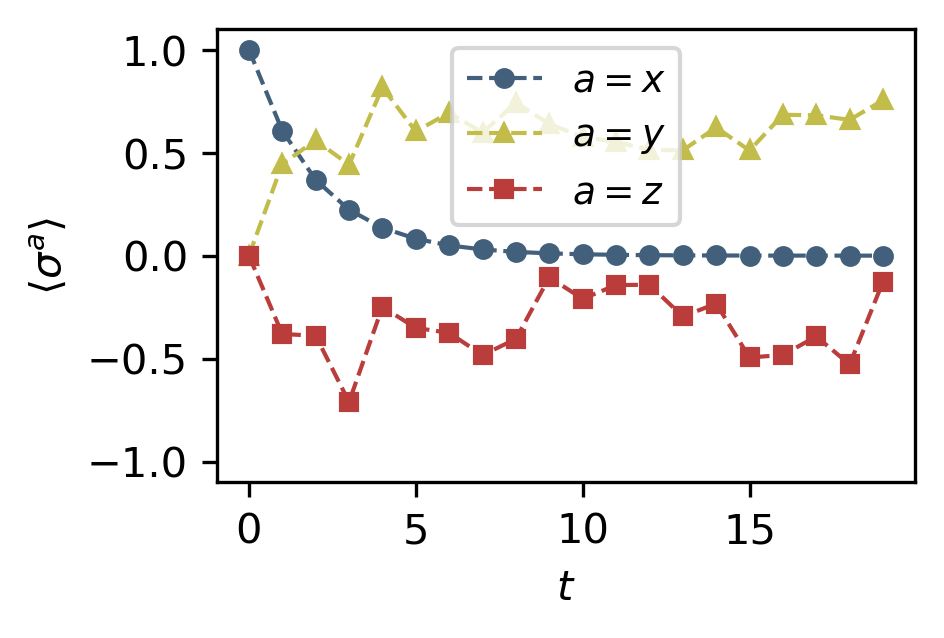}
\caption{Dynamics of the spin projections $\braket{\sigma^a}_t$ for $a=x,y,z$ when driven with a random input sequence with Hamiltonian $H=z_t\frac{h}{2}\sigma^x$. The initial condition is the maximal coherent state $\rho=1/2\sum^1_{i,j=0}\ket{i}\bra{j}$ and the system parameters are $\Delta \tau=1$, $\gamma=1$ and $h=1$.}\label{Fig:good}
\end{center}
\end{figure}

Now we further extend the model to incorporate the measurement formalism described in \cite{mujal2022time}. As the composition of CPTP maps can be described as the product of their matrix representations \cite{wolf2012quantum}, we just need to obtain the Kraus operators of ~\eqref{eq:deph}, which are $K_0=\sqrt{e^{-\frac{g^2}{2}}}I$, $K_1=\sqrt{1-e^{-\frac{g^2}{2}}}\ket{0}\bra{0}$ and 
$K_2=\sqrt{1-e^{-\frac{g^2}{2}}}\ket{1}\bra{1}$, where $\ket{0}$ and $\ket{1}$ are the basis states in the $z$ axis.  The matrix representation of $\mathcal{E}_{\text{deph}}$ in the Pauli matrix basis is then
\begin{equation}
    \widehat{\mathcal{E}}_{\text{deph}}=\left( \begin{matrix} 1 & 0 & 0 & 0  \\ 
   0 & e^{-\frac{g^2}{2}} & 0 & 0 \\
   0 & 0 &  e^{-\frac{g^2}{2}} & 0 \\
   0 &  0 &  0 & 1
  \end{matrix} \right).
\end{equation}
It is straightforward to check that the maximum singular value of $\widehat{\mathcal{E}}_{\text{deph}}$ restricted to the traceless hyperplane is equal to one. The final matrix $\widehat{T}'=\widehat{\mathcal{E}}_{\text{deph}}\widehat{T}$ is
\begin{equation}
  \widehat{T}'=\left( \begin{matrix} 1 & 0 & 0 & 0  \\ 
   0 & e^{-\frac{g^2}{2}}\widehat{T}_{22} & 0 & 0 \\
   e^{-\frac{g^2}{2}}\widehat{T}_{31} & 0 &  e^{-\frac{g^2}{2}}\widehat{T}_{33} & e^{-\frac{g^2}{2}}\widehat{T}_{34} \\
   \widehat{T}_{41} &  0 &  \widehat{T}_{43} & \widehat{T}_{44}
  \end{matrix} \right).
\end{equation}
Given that the maximum singular value of $p(\textbf{z})$ is smaller than one (for $\gamma\neq 4h_t$ and $\gamma\neq 0$), we find that  $\vertiii{p'(\textbf{z})}_2\leq \vertiii{p_{\text{deph}}}_2\cdot \vertiii{p(\textbf{z})}_2<1-\epsilon$ for some $\epsilon>0$ in the 2-Schatten norm, where 
 $p(\textbf{z})$, $p'(\textbf{z})$ and $p_{\text{deph}}$ are the restrictions to the traceless hyperplane of matrices $\widehat{T}$, $\widehat{T}'$ and $\widehat{\mathcal{E}}_{\text{deph}}$ respectively  (given by \eqref{eq:p_ij}). Then, the QRC system has the ESP and the FMP. The single fixed point of the map is given by
\begin{equation}
    \rho^*=\frac{1}{\gamma^2+2h^2_t} \left( \begin{matrix} h^2_t-f_1 & i\gamma h_t(1-f_2) \\ -i\gamma h_t(1-f_2) & \gamma^2 +h^2_t+f_1  \end{matrix} \right),
\end{equation} 
where 
\begin{widetext}
\begin{equation}
\begin{split}
    &f_1 =\frac{4\gamma h^2_t\sinh\left(\frac{g^2}{4}\right)\sinh\left(\frac{\Delta \tau }{4}\sqrt{\gamma^2-16h_t^2}\right)}{\sqrt{\gamma^2-16h_t^2}\left(\cosh\left(\frac{g^2+3\gamma\Delta \tau }{4}\right)-\cosh\left(\frac{g^2}{4}\right)\cosh\left(\frac{\Delta \tau }{4}\sqrt{\gamma^2-16h_t^2}\right)\right)+\gamma\sinh\left(\frac{g^2}{4}\right)\sinh\left(\frac{\Delta \tau }{4}\sqrt{\gamma^2-16h_t^2}\right)}, \\
    &f_2 = \frac{\sinh\left(\frac{g^2}{4}\right)\left(e^{\frac{3\gamma\Delta \tau }{4}}-\cosh\left(\frac{g^2}{4}\right)\cosh\left(\frac{\Delta \tau }{4}\sqrt{\gamma^2-16h_t^2}\right)+\frac{\gamma}{\sqrt{\gamma^2-16h_t^2}}\sinh\left(\frac{g^2}{4}\right)\sinh\left(\frac{\Delta \tau }{4}\sqrt{\gamma^2-16h_t^2}\right)\right)}{\cosh\left(\frac{g^2+3\gamma\Delta \tau }{4}\right)-\cosh\left(\frac{g^2}{4}\right)\cosh\left(\frac{\Delta \tau }{4}\sqrt{\gamma^2-16h_t^2}\right)+\frac{\gamma}{\sqrt{\gamma^2-16h_t^2}}\sinh\left(\frac{g^2}{4}\right)\sinh\left(\frac{\Delta \tau }{4}\sqrt{\gamma^2-16h_t^2}\right)}.
\end{split}
\end{equation}  
\end{widetext}
This fixed point is produced by the competition between the original fixed point in \eqref{eq:fixed_good} ($g\rightarrow 0$) and a diagonal state ($g\rightarrow \infty$). With this we can conclude that this engineered model, even including the measurement protocol, leads to an operational QRC system in the long-term run. 

\section{Conclusions}
\label{sec:conclusions}

In this paper, we have unified the density matrix approach of previous works in QRC with the Bloch vector representation. Moreover, we have shown that these representations are linked by system isomorphisms and that various results concerning the ESP and FMP are independent of the chosen representation. We have also observed that the QRC dynamics in the Bloch vectors representation amounts to that of a state-affine system (SAS) of the type introduced in \cite{grigoryeva2018universal} and for which numerous theoretical results have been established. We have capitalized on this connection to shed some light on fundamental questions in QRC theory in finite dimensions. In particular, we found a necessary and sufficient condition for the ESP and FMP in terms of the existence of an induced norm that bounds the CPTP map for all inputs, determining a guideline for its election. The necessity of this boundedness hypothesis emerges out of the compactness of the input space, which is a common requirement in the RC literature. If the input space is not compact, sufficient conditions can still be found in terms of the weighting sequence \cite{RC9}. Besides, we described common situations in which  QRC systems become useless in long term runs which can be summarized by saying that quantum channels that exhibit input-independent fixed points yield trivial input/output dynamics. 

Our work sets the grounds for further analysis and exploration of the QRC theory. Future work can follow several paths, such as studying the connection between spectral properties of QRC models and their performance in memory and information processing tasks, studying infinite-dimensional quantum reservoirs, or including generalized measurements (positive operator-valued measures) and the effect of a finite number of measurements in the statistics of expected values, given the effect that they imprint in the resources of QRC algorithms \cite{mujal2022time}.

\section{Acknowledgments}
We thank A. Sannia for useful discussion and inspiration for this work, and D. Burgarth and M. Rahaman for answering our questions regarding their work and its relation to our study. We also thank the editor and two anonymous referees for input that has significantly improved the paper. R.M.-P. and J.-P.O. acknowledge partial financial support from the Swiss National Science Foundation (Grant No. 200021 175801/1). 
R.M.-P. acknowledges the Spanish State Research Agency for support through the Severo Ochoa and Mar\'ia de Maeztu Program for Centers and Units of Excellence in R\&D (Grant No. MDM-2017-0711) and through the  QUARESC project (Projects No. 2019-109094GB-C21 and 2019-109094GB-C22/AEI/10.13039/501100011033). Part of this work was funded by MICINN/AEI/FEDER and the University of the Balearic Islands through a predoctoral fellowship (Grant No. MDM-2017-0711-18-1) for R.M.-P. R.M.-P. also is grateful for the hospitality of the Division of Mathematical Sciences of the Nanyang Technological University, where most of these results were obtained.

\end{document}